\documentclass[english,a4paper]{article}
\usepackage{dcolumn}
\usepackage{bm}
\usepackage{graphicx}
\usepackage{amsmath}
\usepackage{amsfonts}
\usepackage{amssymb}
\usepackage{amsthm}
\usepackage{amstext}
\usepackage{amsbsy}
\usepackage{amsopn}
\usepackage{amscd}
\usepackage{amsxtra}

\newtheorem{teo}{Theorem}
\newtheorem{lema}{Lemma}
\newtheorem{prop}{Proposition}
\newtheorem{defi}{Definition}

\newtheorem{rmk}{Remark}

\def\openone{\leavevmode\hbox{\small1\kern-3.3pt\normalsize1}}

\usepackage[usenames,dvipsnames]{color}

\begin{document}
\title{Signatures of physical constraints in rotating rigid bodies}
\author{G. J. Gutierrez Guillen\footnote{Institut de Math\'ematiques de Bourgogne - UMR 5584, CNRS, Universit\'e de Bourgogne, F-21000 Dijon, France}, E. Aldo Arroyo\footnote{Centro de Ci\^{e}ncias Naturais e Humanas, Universidade Federal do ABC, Santo Andr\'{e}, 09210-170 S\~{a}o Paulo, SP, Brazil}, P. Marde\v si\'c\footnote{Institut de Math\'ematiques de Bourgogne - UMR 5584, CNRS, Universit\'e de Bourgogne, F-21000 Dijon, France; University of Zagreb, Faculty of Science, Department of Mathematics, Bijeni\v cka cesta 30, 10000 Zagreb, Croatia.}, D. Sugny\footnote{Laboratoire Interdisciplinaire Carnot de
Bourgogne (ICB), UMR 6303 CNRS-Universit\'e Bourgogne-Franche Comt\'e, 9 Av. A.
Savary, BP 47 870, F-21078 Dijon Cedex, France, dominique.sugny@u-bourgogne.fr}}

\maketitle

\begin{abstract}
We study signatures of physical constraints on free rotations of rigid bodies. We show analytically that the physical or non-physical nature of the moments of inertia of a system can be detected by qualitative changes both in the Montgomery Phase and in the Tennis Racket Effect.
\end{abstract}

\section{Introduction}
Physical dynamical systems governed by ordinary differential equations generally depend on  parameters describing their properties. Such parameters can be found either empirically from a fit of the theoretical evolution to experimental data, or from a purely theoretical approach on the basis of first physical principles. Specific constraints on these parameters can appear as a result of the application of fundamental physical laws. A natural question is then to find clear signatures of such constraints in order to detect in the time evolution of the system the well-defined character of the model system under study. We propose in this paper to study this general problem in the case of free rotational dynamics of a rigid body~\cite{arnoldBook,goldsteinBook,reillyBook,landauBook,cushmanBook}. Its dynamical evolution is described by Euler equations which depend on three moments of inertia which characterize the system mass distribution. We consider the generic situation of an asymmetric rigid body for which the three moments are different from each other. Using the principles of mechanics, it can be shown as described below that the sum of two  moments must be larger than the third moment. The equality can be achieved in the limit of a plane object. Nevertheless, Euler equations make mathematical sense even if the conditions on the moments of inertia are not verified. They are also interesting physically because they describe the dynamics of other systems such as spin 1/2 particles subjected to external electromagnetic fields for which such constraints are not relevant~\cite{cat,vandamme2017b}. Given the rotational dynamics of a rigid body, we therefore seek to detect whether this constraint is satisfied or not by the system. In order to find clear and robust signatures, we analyze the behavior of two geometric properties, namely the Montgomery Phase (MP)~\cite{montgomery1991,natario2010,levi1999,cabrera2007} and the Tennis Racket Effect (TRE)~\cite{chicone1991,vandamme2017,mardesic2020}.

The two geometric effects can be viewed as the two sides of the same coin. The MP is a geometric phase which can be interpreted for rotational dynamics as the analog of Berry phase for quantum systems~\cite{montgomery1991,bohmBook}. The MP measures a specific rotation of the rigid body in a space-fixed frame. When the angular momentum performs a loop in the body-fixed frame, the system rotates by some angle around the fixed direction of the angular momentum in the initial frame. The MP is this angle of rotation, denoted $\Delta\phi$, which depends on the energy of the system characterized as shown below by a parameter $c$. As the name suggests, TRE can be observed with a tennis racket, but also in the rotational dynamics of any asymmetric rigid body. TRE has been recently the subject of different studies both in the classical~\cite{mardesic2020,petrov1993} and quantum domains~\cite{vandamme2017b,ma2020,naturephys,hamraoui2018,RMPsugny,opatrny2018}. It is manifested by an almost $\pi$- flip of the angular momentum when the rigid body performs a full rotation around its unstable axis of rotation~\cite{chicone1991,vandamme2017}. The TRE is generally not perfect and the twist defect $\varepsilon$ is a function of the energy, i.e. of the parameter $c$. The physical or non-physical nature of the rigid body can be detected in the global behavior of the two functions $\Delta\phi(c)$ and $\varepsilon(c)$. In this study, a classical system is said to be physical if its dynamic is governed by the Euler equations, while satisfying the constraints on the moments of inertia. The role of the latter in the quantum domain is discussed in the conclusion of the paper~\cite{RMPsugny}. For physical systems and oscillating trajectories for which $c<0$, we show rigorously in this study that $\Delta\phi$ is a decreasing function larger than $2\pi$, while $\varepsilon$ is an injective function. These two properties are not verified if the physical constraints are not met. On the basis of these observations, different experiments can be considered to test the physical nature of the rigid body~\cite{phone2021}. Finally, we take advantage of this global study to extend the proof of Ref.~\cite{mardesic2020} on the existence of the TRE. It was shown in~\cite{mardesic2020} that a perfect TRE can be performed for a trajectory close to the separatrix between the oscillating and rotating motions of the rigid body in the limit of a perfect asymmetric body, in particular without satisfying the physical restriction on the moments of inertia. We prove in this study in which conditions an exact TRE is observed within such constraints.

The paper is organized as follows. In Sec.~\ref{sec.const}, we recall how to derive
the physical constraints on the moments of inertia. The description of the rotational dynamics of an asymmetric rigid body is summarized in Sec.~\ref{sec:descrip}. Section~\ref{sec.Montg} focuses on the dynamical signatures for the MP, while Sec.~\ref{sec:TRE} and~\ref{secphysnonphys} describe the case of the TRE. The global behavior of the TRE is investigated in Sec.~\ref{sec.parameters}. Conclusion and prospective views are given in Sec.~\ref{sec:conc}. Technical results are reported in~\ref{appa} and~\ref{ap.GIE}.

\section{Physical constraints on the moments of inertia}\label{sec.const}
A free rotation of a rigid body can be described by the position of the body-fixed frame with respect to the space-fixed frame, given respectively by the set of coordinates $(x,y,z)$ and $(X,Y,Z)$~\cite{goldsteinBook}. In the body-fixed frame, the angular momentum of the body $\bf{J}$ is connected to the angular velocity ${\bf \omega}$ through the relation ${\bf J}=I \omega$ where $I$ is a $3\times 3$ symmetric matrix, called the inertia matrix. Its eigenvalues are the inertia moments denoted $I_x$, $I_y$ and $I_z$ and corresponding to the three axes of the body-fixed frame. For any rigid body there exist the following physical restrictions on the values of the moments of inertia
\begin{equation}\label{eq.ret}
I_i+I_j \geq I_k
\end{equation}
with $\{i,j,k\}=\{x,y,z\}$ and not equal. This constraint can be established from the definition of the moments of inertia
\begin{equation*}
    I_i =\int_V\rho({\bf r})(x_j^2+x_k^2)d{\bf r},
\end{equation*}
where $\rho$ is the mass density, $V$ the volume of the body and $x_j$ the coordinates of the position vector ${\bf r}$, $(x_1,x_2,x_3)\equiv (x,y,z)$. We deduce that
\begin{equation*}
    I_i=\int\rho({\bf r})(x_j^2+x_k^2)d{\bf r}\leq \int\rho({\bf r})(2x_i^2+x_j^2+x_k^2)d{\bf r}=I_j+I_k .
\end{equation*}
In the case of an asymmetric rigid body such that
\begin{equation}\label{desi2}
 I_z < I_y < I_x,
\end{equation}
the only constraint to satisfy is
\begin{equation}\label{desi3}
I_y + I_z \geq I_x.
\end{equation}
\begin{defi}
Let $I_x$, $I_y$, $I_z$ be the moments of inertia of an asymmetric rigid body such that $I_z < I_y < I_x$. A rigid body is said to be \emph{physical} if the values of the moments of inertia fulfill Eq.~\eqref{desi3}.
\end{defi}
In \cite{mardesic2020}, two parameters $a$ and $b$ describing the asymmetry of the body were defined as
\begin{equation}\label{rela1}
a = \frac{I_y}{I_z} - 1  \; , \;\;\;\;\; b = 1 -\frac{I_y}{I_x}.
\end{equation}
Note that, by definition, $a>1$ and $0<b<1$. The physical constraint~\eqref{desi3} can be easily written by introducing a constant $\mathcal{I}=\mathcal{I}(a,b)$, which we call the \emph{geometric constant}. For further use, we also introduce a \emph{second geometric constant} $\mathcal{J}=\mathcal{J}(a,b)$. We have
\begin{equation}\label{eq:I,J}
\mathcal{I}=\frac{1-b}{\sqrt{b(a+b)}},\quad
 \mathcal{J}=  \frac{a+1}{\sqrt{a
(a+b)}} - 1.
\end{equation}


\begin{prop}\label{prop.physical}
Let $I_x$, $I_y$, $I_z$ be the moments of inertia of an asymmetric rigid body such that $I_z < I_y < I_x$. The moments of inertia $I_x$, $I_y$, $I_z$ describe a physical rigid body if and only if the geometric constant $\mathcal{I}$  satisfies the inequality $\mathcal{I}\geq 1$.
\end{prop}
\begin{proof}
From Eq.~(\ref{rela1}), we get
\begin{equation*}\label{rela2}
I_y = (1-b)I_x \;, \;\;\;\;\;\; I_z =\Big(\frac{1-b}{1+a}\Big) I_x.
\end{equation*}
Since $I_x\neq 0$, we have $I_y + I_z \geq I_x$, if and only if
\begin{equation*}\label{desi4}
(1-b) + \Big(\frac{1-b}{1+a}\Big) \geq 1 .
\end{equation*}
Using $0<1+a$, Eq.~\eqref{desi4} is verified if and only if
\begin{equation*}
(1-b)(1+a)+(1-b)\geq 1+a,
\end{equation*}
which is equivalent to $ab+b^2\leq 1-2b+b^2$ and finally to $\mathcal{I}\geq1$.
\end{proof}
The set of points of coordinates $(a,b)$ that fulfill the constraint $\mathcal{I}\geq 1$ is represented in Fig.~\ref{figur1}.
\begin{figure}[h!]
  \centering
  \includegraphics[width=0.8\linewidth]{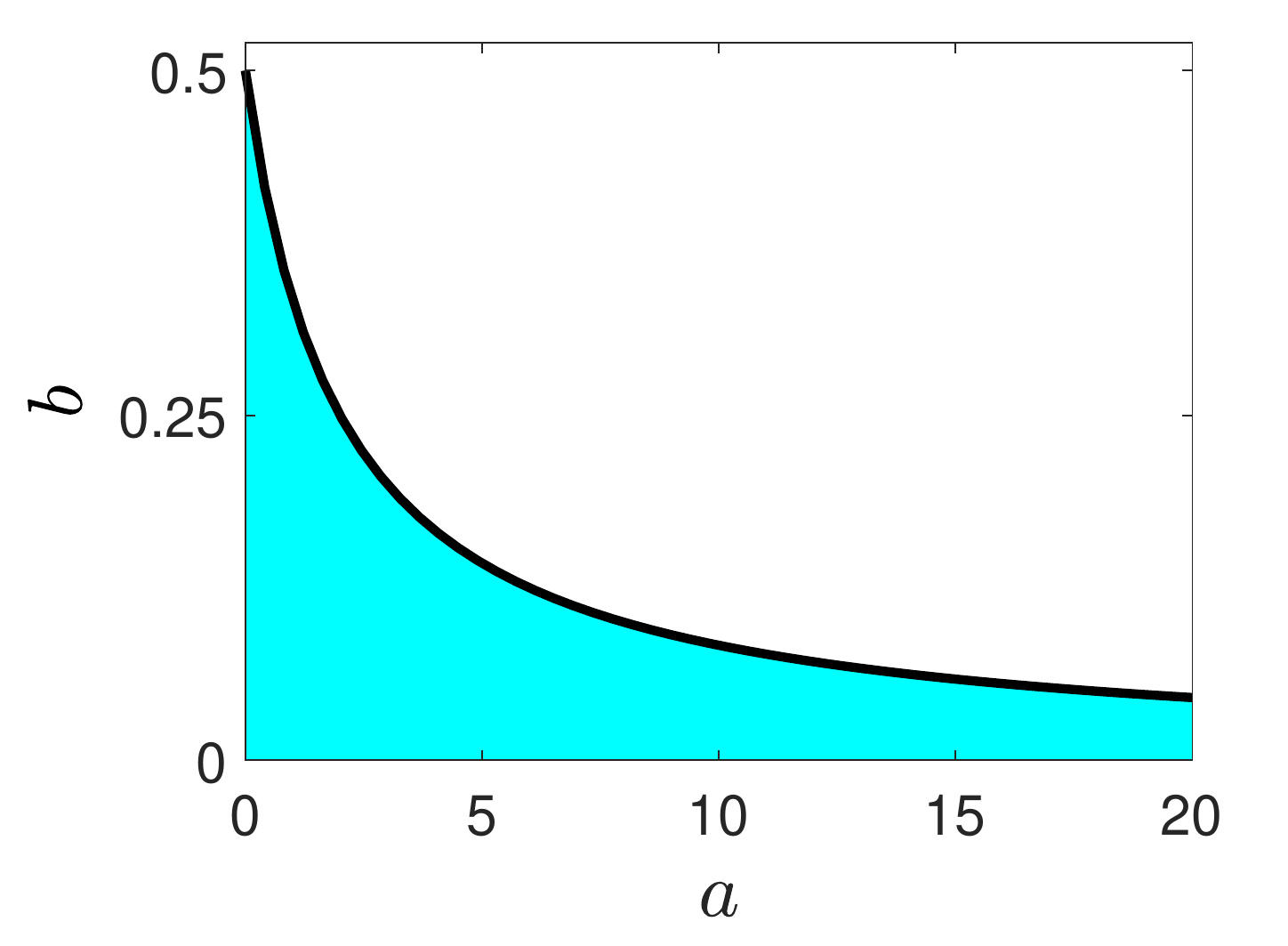}
    \caption{Set of points $(a,b)$ (blue area) such that $\mathcal{I}=\frac{1-b}{\sqrt{b(a+b)}} \geq 1$. The solid black line represents the points for which $\mathcal{I}=1$.}
  \label{figur1}
\end{figure}

\section{Description of the dynamical system}\label{sec:descrip}
We recall in this section how to describe the rotational dynamics of a rigid body~\cite{arnoldBook,goldsteinBook}. The relative motion of the body-fixed frame with respect to the space-fixed frame can be described by the three Euler angles $(\theta,\psi,\phi)$. This set of angles is not unique and we use in this work the same definition as in~\cite{mardesic2020}, which is well suited to study the MP and the TRE. The angular momentum ${\bf J}$ of the system is a constant of motion in the space-fixed frame, and assumed to be along the $Z$-axis by convention. It is not the case in the body-fixed frame even if its modulus, $J=|{\bf J}|$, does not depend on time. The two angles $(\theta,\psi)$ describe the position of ${\bf J}$ as $J_x=-J\sin\theta\cos\psi$, $J_y=J\sin\theta\sin\psi$ and $J_z=J\cos\theta$. We set below $J=1$ without loss of generality. The dynamics of the Euler angles are governed by the following differential equations~\cite{mardesic2020}
\begin{eqnarray}\label{eqeulerangle}
\dot{\theta}&=&(\frac{1}{I_y}-\frac{1}{I_x})\sin\theta\sin\psi\cos\psi \nonumber \\
\dot{\phi}&=&\frac{1}{I_y}\sin^2\psi+\frac{1}{I_x}\cos^2\psi \\
\dot{\psi}&=&(\frac{1}{I_z}-\frac{1}{I_y}\sin^2\psi-\frac{1}{I_x}\cos^2\psi)\cos\theta. \nonumber
\end{eqnarray}
This dynamical system has two constants of motion $H$ and $J$, where $H=\frac{J_x^2}{2I_x}+\frac{J_y^2}{2I_y}+\frac{J_z^2}{2I_z}$, and therefore defines a classical Hamiltonian integrable dynamic. The phase space of an asymmetric rigid body has a relatively simple structure made of a separatrix for $H=\frac{J^2}{2I_y}$ which delimits the oscillating and rotating trajectories, organized each around a stable fixed point in the angular momentum space for $H=\frac{J^2}{2I_x}$ and $H=\frac{J^2}{2I_z}$, respectively. We introduce the parameter $c=2I_yH-1$ which represents the signed distance to the separatrix and verifies $-b\leq c\leq a$. The MP and the TRE are two geometric properties which do not depend on time. It is thus interesting to introduce the relative motion of $\psi$ and $\theta$ with respect to $\phi$. This is done using Eq.~\eqref{eqeulerangle}. Replacing the terms $\sin\theta$ and $\cos\theta$ using the formula
\begin{equation}\label{eq:c}
    c=a-\sin^2\theta(a+b\cos^2\psi),
\end{equation}
one gets
\begin{eqnarray}\label{eqred}
& & \frac{d\theta}{d\phi}=\pm \frac{b\sqrt{a-c}\sin\psi\cos\psi}{(1-b\cos^2\psi)\sqrt{a+b\cos^2\psi}} \label{eqred1}\\
& & \frac{d\psi}{d\phi}=\pm \frac{\sqrt{(a+b\cos^2\psi)(c+b\cos^2\psi)}}{1-b\cos^2\psi}. \label{eqred2}
\end{eqnarray}
A shematic representation of the reduced phase space $(\psi,\frac{d\psi}{d \phi})$ is given in Fig.~\ref{fignew}. Note that the sign of $\sin\theta$ and $\cos\theta$ is not fixed by Eq.~\eqref{eq:c}, which leads to the two possible expressions of the derivatives (corresponding to the $\pm$ sign) in Eq.~\eqref{eqred1} and \eqref{eqred2}. The right choice of this sign will play a crucial role in the description of the geometric properties.

\begin{figure}[h!]
  \centering
  \includegraphics{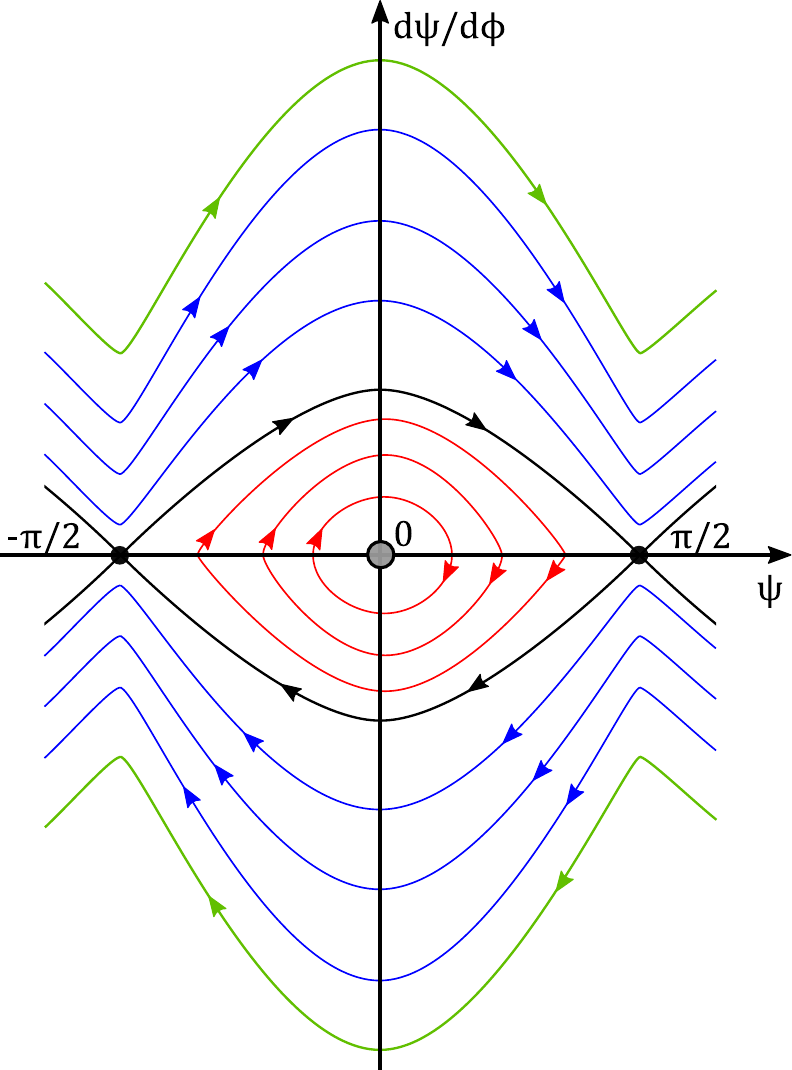}
  \caption{Schematic representation of the reduced phase space $(\psi,\frac{d\psi}{d \phi})$. The red, black and blue curves depict respectively the oscillating trajectories, the separatrix and the rotating trajectories. The grey dot and the green curves indicate respectively the position of the fixed points for the oscillating and rotating motions (\textit{i.e.} the values $c=-b$ and $c=a$ respectively). The same color code is used through the paper.}
  \label{fignew}
\end{figure}

\section{The Montgomery phase}\label{sec.Montg}
We study in this section the signature of the physical constraint introduced in Prop.~\ref{prop.physical} on the MP. Different analytical properties described by Th.~\ref{th:MP} and~\ref{th:MProt} are found for oscillating and rotating trajectories.
\begin{teo}\label{th:MP} \;
\begin{enumerate}
    \item
    For any $(a,b)\in (0,\infty)\times(0,1)$, the greatest lower bound $\inf_{c\in(-b,0)} \Delta\phi(a,b)$ of the Montgomery phase
    $\Delta\phi$, for oscillating trajectories, is given by
    $$\inf_{c\in(-b,0)} \Delta\phi(a,b)=2\pi\mathcal{I},$$
    where $\mathcal{I}=\mathcal{I}(a,b)$ is the geometric constant given by \eqref{eq:I,J}.
\item A rigid body is physical (i.e. such that $\mathcal{I}\geq1$), if and only if $$\inf_{c\in(-b,0)}\Delta \phi(a,b)  \geq 2\pi.$$
\end{enumerate}
\end{teo}
\begin{teo}\label{th:MProt}
    For any $(a,b)\in(0,\infty)\times(0,1)$, the greatest lower bound $\inf_{c\in(0,a)} \Delta\phi(a,b)$ of the Montgomery phase
    $\Delta\phi(a,b)$, for rotating trajectories is given by
    $$\inf_{c\in(0,a)} \Delta\phi(a,b)=2\pi\mathcal{J},$$
    where $\mathcal{J}=\mathcal{J}(a,b)$ is the second geometric constant given by \eqref{eq:I,J}.
\end{teo}
The corresponding proofs are given below. We observe that Th.~\ref{th:MP} gives a direct way to detect the physical nature of the rigid body. The evolution of $\Delta\phi$ with respect to $c$ in the oscillating and rotating cases is presented in Fig.~\ref{figmong} for two generic pairs $(a,b)$ corresponding to a physical and a non-physical rigid body.

\begin{figure}[h!]
  \centering
  \includegraphics[width=0.6\textwidth]{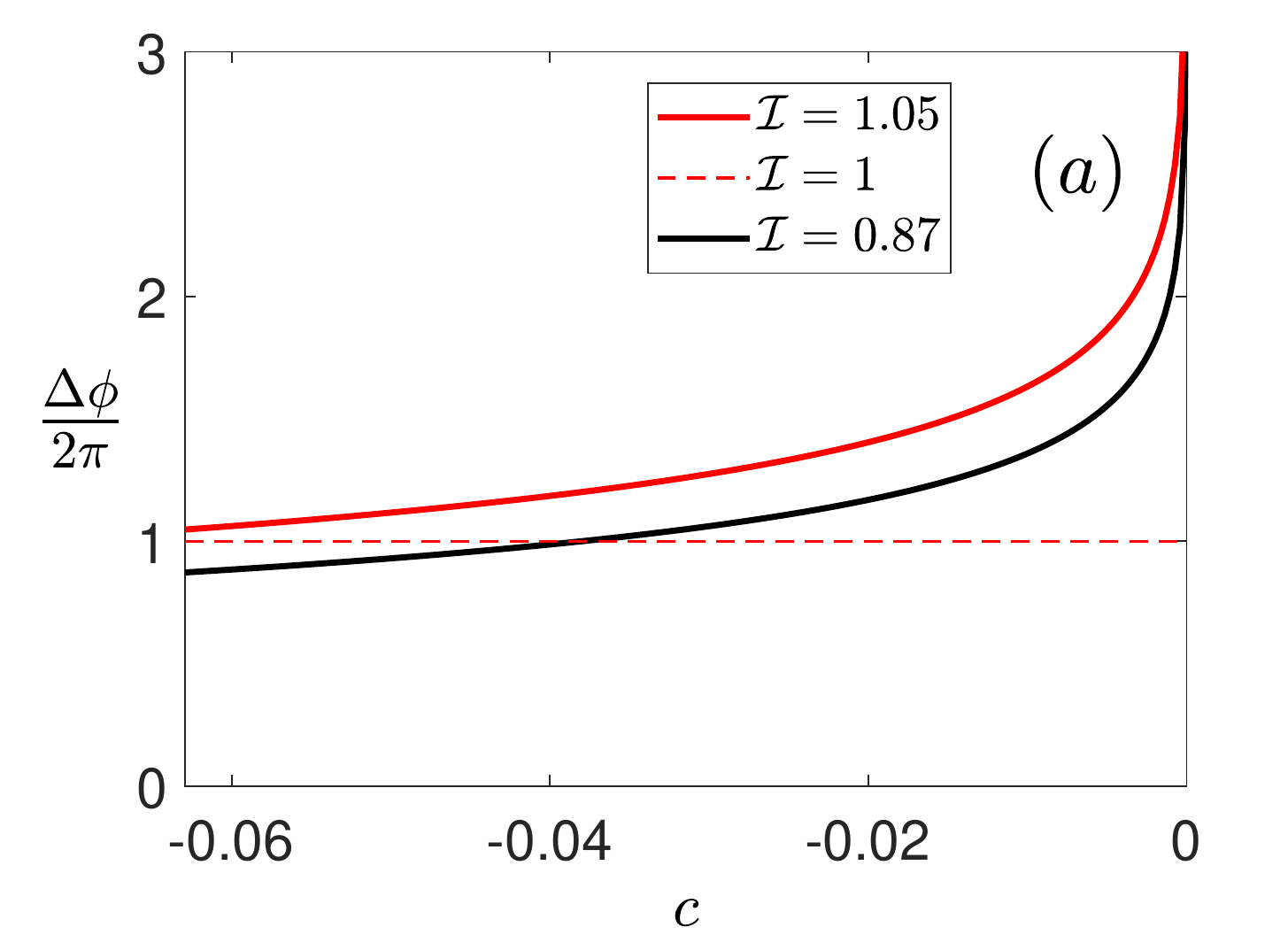}
  \includegraphics[width=0.6\textwidth]{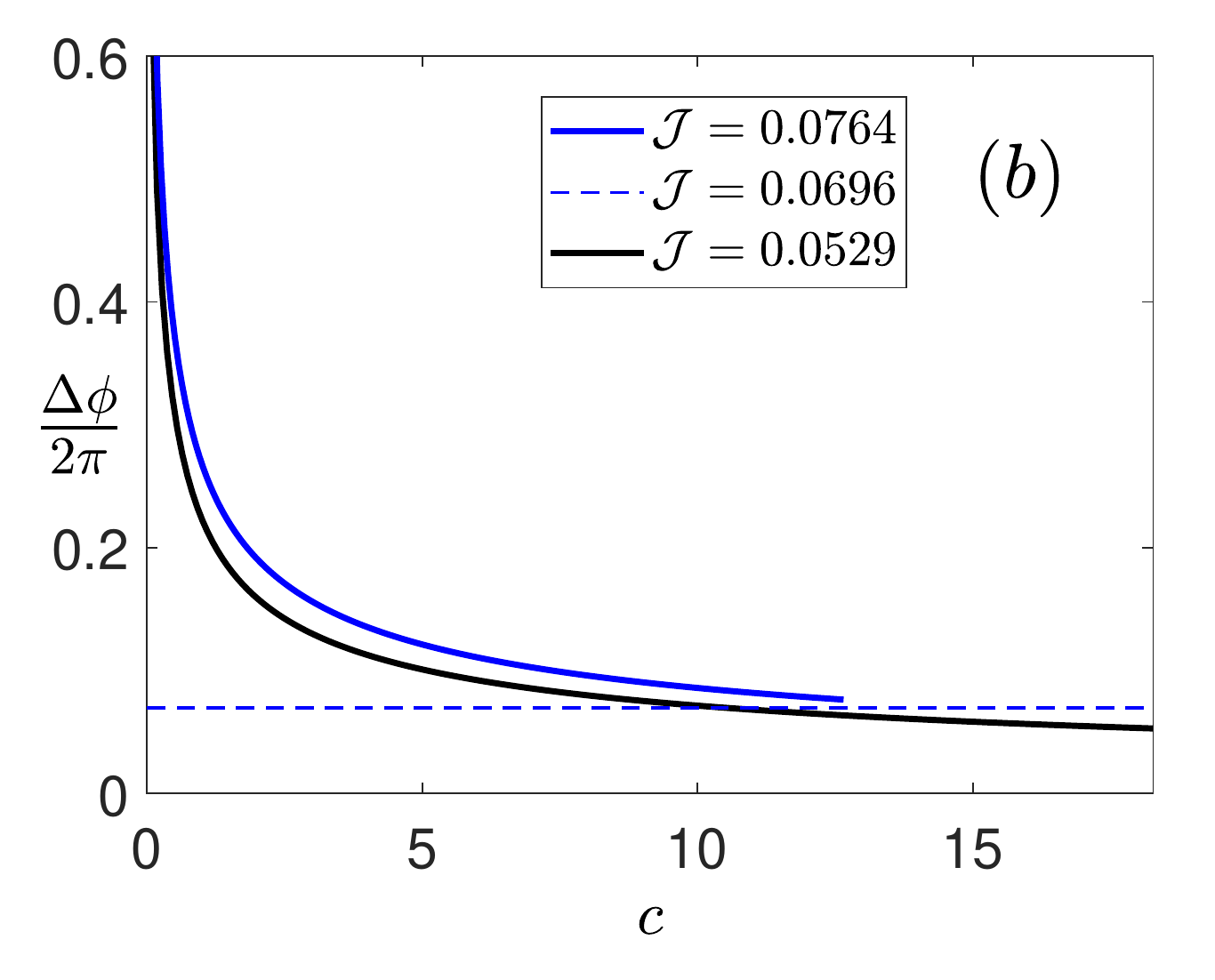}
  \caption{Plot of the Montgomery phase $\Delta\phi$ as a function of the parameter $c$ for oscillating (panel (a)) and rotating (panel (b)) trajectories. A physical and a non-physical rigid bodies are considered with respectively the parameters $(a=12.65,b=0.0629)$ and $(a=18.27,b=0.0629)$. Note that the parameter $c$ belongs to the interval $[-b,a]$. The physical and non-physical cases are depicted respectively in red or blue and in black. The horizontal dashed line represents the boundary between the physical and the non-physical bodies. The small inserts give the corresponding values of $\mathcal{I}$ and $\mathcal{J}$.}
  \label{figmong}
\end{figure}

\section*{The case of oscillating trajectories}\label{sub.oscillating}
The MP is defined as the variation of the angle $\phi$ for a loop in the reduced phase space  $(\psi,\frac{d \psi}{d \phi})$ as shown in Fig.~\ref{fignew}. For the oscillating trajectories, the condition $c+b\cos^2\psi\geq 0$ bounds the evolution of $\psi$ and leads to $\sin^2\varepsilon\geq \frac{|c|}{b}$. We denote by $\varepsilon^*=\arcsin{\sqrt{\frac{-c}{b}}}$ this minimal value. Using the symmetry of the trajectory with respect to $\frac{d \psi}{d \phi}=0$, the variation $\Delta \phi$ can be expressed as
\begin{align}\label{mongo1}
\Delta \phi = 2 \int_{-\frac{\pi
}{2}+\varepsilon^*}^{\frac{\pi
}{2}-\varepsilon^*} \frac{1-b \cos^2 \psi
}{\sqrt{(a+b\cos^2 \psi )(c+b\cos^2 \psi)}}d \psi .
\end{align}


\begin{proof}[Proof of Theorem~\ref{th:MP}]
The integral~(\ref{mongo1}) is also symmetric with respect to
the line defined by $\psi=0$, so that
\begin{align}
\label{mongo2} \Delta \phi = 4 \int_{0}^{\frac{\pi
}{2}-\arcsin(\sqrt{\frac{-c}{b}})} \frac{1-b \cos^2 \psi
}{\sqrt{(a+b\cos^2 \psi )(c+b\cos^2 \psi)}}d \psi .
\end{align}
Performing the change of variables $x=\cos^2 \psi$, we get:
\begin{align}
\label{mon3} \Delta \phi =  \frac{2}{\sqrt{ab}}\int_{-c/b}^1
\frac{1-bx}{\sqrt{x(1-x)(1+\frac{b}{a}x)(x+\frac{c}{b})}}dx .
\end{align}
Introducing $u = -c/b$, so that for $c \in (-b,0)$, $u \in (0,1)$, we get
\begin{align}
\label{mon4} \Delta \phi = 2 F_{a,b}(-b u ,u) ,
\end{align}
where  the function $F_{a,b}(-b u ,u)$ described in~\ref{appa} is the incomplete elliptic integral
\begin{align}
\label{mon5} F_{a,b}(-b u ,u) = \frac{1}{\sqrt{ab}}\int_u^1
\frac{1-bx}{\sqrt{x(1-x)(1+\frac{b}{a}x)(x-u)}}dx.
\end{align}
According to Lemma~\ref{lema.decreasing} of~\ref{appa}, giving the monotonic behavior of the function $F_{a,b}(-bu,u)$, we know that the greatest lower bound of the
function $F_{a,b}(-b u ,u)$ occurs when $u \rightarrow 1$ which is computed in Lemma~\ref{lema.F}, namely
\begin{align}
\label{mon6} \lim_{u \rightarrow 1} F_{a,b}(-b u ,u)
=\frac{1-b}{\sqrt{b(a+b)}} \pi=\pi \mathcal{I},
\end{align}
thus showing the first claim. It is then straightforward to show the second statement.
\end{proof}

\section*{The case of rotating trajectories}\label{sub.Rotating}
For rotating trajectories, the angular momentum performs a loop when the angle $\psi$ goes from $-\frac{\pi}{2}$ to $\frac{3\pi}{2}$. Using the symmetries of $\cos^2(\psi)$, we obtain that the corresponding variation of $\phi$ is given by
\begin{align}
\label{mongoro1} \Delta \phi = 4 \int_{0}^{\frac{\pi }{2}}
\frac{1-b \cos^2 \psi }{\sqrt{(a+b\cos^2 \psi )(c+b\cos^2 \psi)}}d
\psi .
\end{align}

\begin{proof}[Proof of Theorem~\ref{th:MProt}]
The proof follows the same general lines as the proof of Th.~\ref{th:MP}. One shows that the function $c\mapsto\Delta\phi(c)$ is decreasing. Hence, since $0\leq c < a$, the greatest lower bound is obtained for  $c\to a$. We obtain by direct calculations,
\begin{align}
\label{mongoro3} \lim_{c \rightarrow a}\Delta \phi &= 4
\int_{0}^{\frac{\pi }{2}} \frac{1-b \cos^2 \psi }{a+b\cos^2 \psi}d
\psi =
4 \int_{0}^{\frac{\pi }{2}} \Big[ \frac{1+a}{a+b\cos^2
\psi} - 1 \Big] d \psi \nonumber \\
&=2 \pi  \Big[ \frac{a+1}{\sqrt{a (a+b)}}
- 1 \Big]=2\pi\mathcal{J}.
\end{align}
\end{proof}


\section{The tennis racket effect}\label{sec:TRE}
In previous works~\cite{chicone1991,vandamme2017,mardesic2020}, a study of the TRE has been made in a neighborhood of the separatrix. The corresponding trajectories are assumed to exhibit the TRE since the separatix is the curve connecting the two unstable points, which are at the origin of this geometric effect~\cite{arnoldBook,chicone1991}. In short, the TRE is characterized by a variation $\Delta\phi=2\pi$ when $\Delta\psi=\pi-2\varepsilon$. A perfect TRE is described by $\varepsilon=0$, $\varepsilon$ describing the defect to a perfect TRE. The pair of values $(c,\varepsilon)$ defines a curve denoted $\mathcal{C}$.  In addition to the signatures that the physical constraint imposes on the TRE, we take advantage of this analysis to make a global study of the TRE. A local, versus global, study of the TRE corresponds therefore to a local (close to $c=0$), or  global study of the curve $\mathcal{C}$, respectively.

We study the dynamics of the reduced system
\begin{align}
\frac{d\psi}{d\phi}&=\frac{(a+b\cos^2\psi)\cos\theta}{1-b\cos^2\psi}\label{eq.psiphi}\\
\frac{d\theta}{d\phi}&=\frac{b\sin\theta\sin\psi\cos\psi}{1-b\cos^2\psi}\label{eq.thetaphi} ,
\end{align}
with first integral $c=a-\sin^2\theta(a+b\cos^2\psi)$. A reduced phase space of this system is represented in Fig.~\ref{fig.PS}. In this reduced phase space, the evolution of the angle $\phi$ can also be analyzed since, in this system, $\phi$ plays the role of time.

\begin{figure}[h!]
  \centering
  \includegraphics{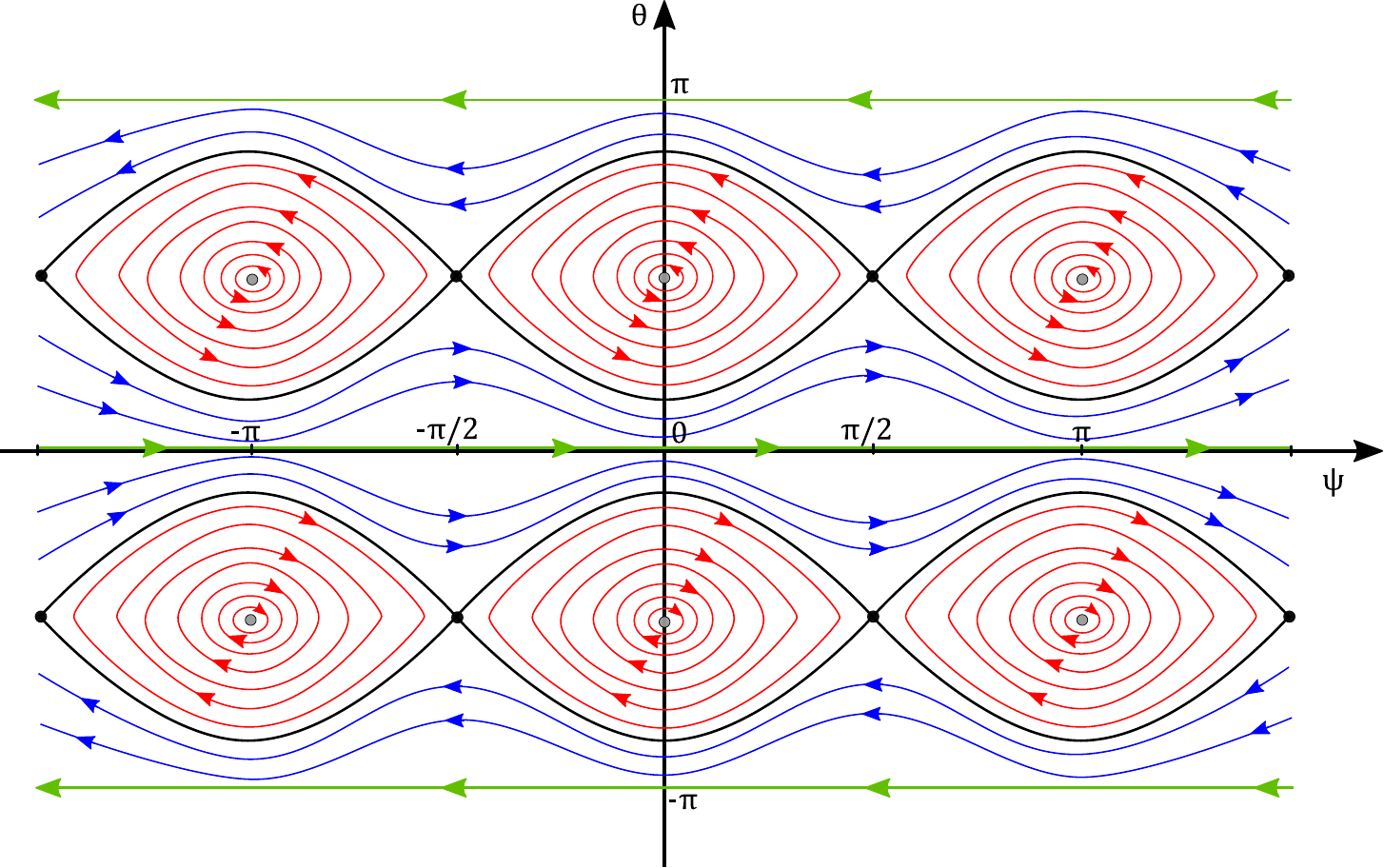}
    \caption{Schematic representation of the phase space $(\psi,\theta)$. The blue and red lines correspond respectively to the rotating and oscillating trajectories. The separatrix is plotted in black ($c=0$). The grey dots and the green curves have the values $c=-b$ and $c=a$ respectively.\label{fig.PS}}
\end{figure}
We define the curve $\mathcal{C}$ as the set of pairs $(c,\varepsilon)$ such that when considering symmetric initial and final values of $\psi$, $\psi_i=-\frac{\pi}{2}+\varepsilon$ and $\psi_f=\frac{\pi}{2}-\varepsilon$,  for the trajectory given by the value of $c$, the time taken to go from the initial to the final point along the trajectory is $2\pi$ (\textit{i.e.} $\Delta\phi=2\pi$). As already mentioned, Eq.~\eqref{eq.psiphi} can be used to study this effect. A difficulty lies in the fact that the sign of the term $\cos\theta$ given by
\begin{equation}\label{eq.cos}
\cos\theta=\pm\sqrt{\frac{c+b\cos^2\psi}{a+b\cos^2\psi}} ,
\end{equation}
and then of the derivative $d\psi/d\phi$ depend on the region of the phase space we consider. To detect this change of sign, we introduce the curve
\begin{equation*}
    \mathcal{S}=\{(c,\varepsilon)\mid c= -b\sin^2(\varepsilon)\}
\end{equation*}
for which $\cos\theta$ is equal to 0 and $\varepsilon$ is determined by the initial condition $\psi=-\pi/2+\varepsilon$. A change of sign corresponds therefore to an intersection point between the curves $\mathcal{C}$ and $\mathcal{S}$.

In this section, we prove that the curve $\mathcal{C}$ is given as a graph of a function $\varepsilon=\varepsilon(c)$. We derive implicit equations describing this function and we study different properties of the curve leading to a global description of TRE. The main result about the signature of the physical constraint on the TRE can be stated as follows.

\begin{teo}\label{teo.Rinj}
The function $\varepsilon(c)$ is injective, if and only if the geometric constant $\mathcal{I}$ verifies $\mathcal{I}\geq1$,
that is if and only if the rigid body is physical.
\end{teo}
The proof of Th.~\ref{teo.Rinj} can be found in Sec.~\ref{sub.phy}.

\section*{Description of the curve $\mathcal{C}$}
We first establish different results about the set $\mathcal{C}$.
\begin{lema}
The set $\mathcal{C}$ is a curve in the plane $(c,\varepsilon)$ given as a graph of a function
$c\mapsto \varepsilon(c)$.
\end{lema}
\begin{proof}
We need to prove that for a fixed value of $c$ there exists only one  value of $\varepsilon$. Let us assume that there exist two values of $\varepsilon$, namely $\varepsilon_1$ and $\varepsilon_2$ fulfilling the definition of the curve $\mathcal{C}$. The time taken to go from $-\frac{\pi}{2}+\varepsilon_k$ to $\frac{\pi}{2}-\varepsilon_k$ along the curve is equal to $2\pi$ for $k=1,2$. We deduce that the time needed to go from one point to the other (for example between $-\frac{\pi}{2}+\varepsilon_1$ and $-\frac{\pi}{2}+\varepsilon_2$) is 0. We conclude that $\varepsilon_1=\varepsilon_2$.
\end{proof}

\subsection*{First implicit equation}
We consider the derivative $\frac{d\psi}{d\phi}$ for $\theta\in[0,\frac{\pi}{2})$, i.e. in the case $\cos\theta>0$. Using Eqs.~\eqref{eq.psiphi} and \eqref{eq.cos}, one obtains that the TRE can be described as the solutions $(c,\varepsilon)$ of the implicit equation
\begin{equation}\label{eq.1}
2\pi=\int_{-\pi/2+\varepsilon}^{\pi/2-\varepsilon} \frac{1-b\cos^2(\psi)}{\sqrt{(b\cos^2(\psi)+a)(b\cos^2(\psi)+c)}}d\psi .
\end{equation}
The square root in the integrand of Eq.~\eqref{eq.1} is well defined if $c\geq -b\sin^2(\varepsilon)$, \textit{i.e.} before intersecting the curve $\mathcal{S}$.

For non-negative values of $\varepsilon$, one can consider the change of coordinates given by $x=\cos^2(\psi)$, which leads to
\begin{equation}\label{eq1}
2\pi = F_{a,b}(c,u)=\frac{1}{\sqrt{ab}}\int_u^1 \frac{1-bx}{\sqrt{x(1-x)(1+\frac{b}{a}x)(x+\frac{c}{b})}}dx,
\end{equation}
where $u=\sin^2\varepsilon$, $a$ and $b$ are fixed parameters and $c$ and $u$ the variables. Note that $u\geq 0$ and $\mathcal{S}=\{(c,u)\mid c= -bu\}$.

We then establish some properties of $\mathcal{C}$ in the region where the curve is described by Eq.~\eqref{eq.1}.
\begin{prop}
    The function $\varepsilon(c)$ given by the solutions of the implicit Eq.~\eqref{eq.1} is an injective decreasing function.
\end{prop}
\begin{proof}
    Using the symmetry of the integrand, we can transform Eq.~\eqref{eq.1} into
\begin{equation}\label{eq.2}
\pi=I_{a,b}(c,\varepsilon)=\int_{0}^{\pi/2-\varepsilon} \frac{1-b\cos^2(\psi)}{\sqrt{(b\cos^2(\psi)+a)(b\cos^2(\psi)+c)}}d\psi .
\end{equation}

We analyze the set $\mathcal{C}$ near a point $(c_1,\varepsilon_1)$, for which the curve is described by Eq.~\eqref{eq.2}. The function given by $I_{a,b}(c,\varepsilon)- \pi$ is a continuously differentiable function such that $I_{a,b}(c_1,\varepsilon_1)-\pi=0$. Thus, we can use the Implicit Function Theorem to compute the derivative of the function $\varepsilon(c)$. The partial derivative of the function with respect to $\varepsilon$ is
\begin{equation*}
\frac{\partial(I_{a,b}(c,\varepsilon)- \pi)}{\partial \varepsilon}(c_1,\varepsilon_1)=-\frac{1-b\cos^2(\pi/2-\varepsilon_1)}{\sqrt{(b\cos^2(\pi/2-\varepsilon_1)+a)(b\cos^2(\pi/2-\varepsilon_1)+c_1)}}\neq 0 .
\end{equation*}
This partial derivative is not equal to 0, thus the Implicit Function Theorem gives the existence of a local solution near the point $(c_1,\varepsilon_1)$ and leads to the following expression for the derivative
\begin{equation*}
\left.\frac{d\varepsilon}{dc}=-\left({\frac{\partial I_{a,b}}{\partial c}}\right)\bigg/\left({\frac{\partial I_{a,b}}{\partial \varepsilon}}\right) \right|_{\varepsilon(c)}.
\end{equation*}
Using the definition of $I_{a,b}$, we obtain
\begin{equation}\label{eq.deri}
\frac{d\varepsilon}{dc}=-\frac{\int\limits_{0}^{\pi/2-\varepsilon(c)} \frac{1-b\cos^2(\psi)}{\sqrt{b\cos^2(\psi)+a}}\left(\frac{1}{2(b\cos^2(\psi)+c)^\frac{3}{2}}\right)d\psi}{\frac{1-b\cos^2(\pi/2-\varepsilon(c))}{\sqrt{(b\cos^2(\pi/2-\varepsilon(c))+a)(b\cos^2(\pi/2-\varepsilon(c))+c)}}} \leq 0.
\end{equation}
Since Eq.~\eqref{eq.deri} is valid near any point for which the curve $\mathcal{C}$ is described by the first implicit equation, we conclude that this part of the curve is strictly decreasing (the expression~\eqref{eq.deri} has a discrete set of zeros). The injectivity follows form this property.
\end{proof}
Equation~\eqref{eq1} can be used to study the existence of a TRE on the separatrix and a perfect TRE as shown by the following result.

\begin{prop}\label{prop.perf} \;
    \begin{enumerate}
        \item The TRE always occurs on the separatrix.
        \item A perfect TRE occurs if and only if $\mathcal{J}<2.$
    \end{enumerate}
\end{prop}
\begin{proof}
    \begin{enumerate}
        \item Having a TRE on the separatrix is equivalent as finding $u_0=\sin^2(\varepsilon_0)$ such that
\begin{equation}\label{eqsepara}
2\pi = F_{a,b}(0,u_0)=\frac{1}{\sqrt{ab}}\int_{u_0}^1 \frac{1-bx}{x\sqrt{(1-x)(1+\frac{b}{a}x)}}dx.
\end{equation}
Using the expression on the right-hand side of Eq.~\eqref{eqsepara}, we observe that $F_{a,b}(0,u)$ goes respectively to $+\infty$ and $0$ when $u$ goes to 0 and 1. Thus we conclude that there exists $u_0\in(0,1)$ such that $2\pi = F_{a,b}(0,u_0)$.
        \item Showing the existence of a perfect TRE ($\varepsilon=0$) amounts to show the  existence of $c_0\in(-b,a)$ such that
\begin{equation*}
2\pi = F_{a,b}(c_0,0)=\frac{1}{\sqrt{ab}}\int_{0}^1 \frac{1-bx}{\sqrt{x(1-x)(1+\frac{b}{a}x)(x+\frac{c_0}{b})}}dx.
\end{equation*}
We have that $F_{a,b}(c,0)$ goes to $\infty$, when $c$ goes to 0, and following the computations of Sec.~\ref{sub.Rotating}, we know that $2F_{a,b}(c,0)$ is a decreasing function with a greatest lower bound equal to $2\pi\mathcal{J}$. We conclude that a perfect TRE occurs if and only if $\mathcal{J}<2$.
    \end{enumerate}
\end{proof}
The values $\varepsilon_0$ and $c_0$ represent the TRE on the separatrix and a perfect TRE respectively. A schematic representation is given in Fig.~\ref{fig.Solution}. In Sec.~\ref{sec.parameters}, we study the behavior of these effects when  the parameter $a$ varies.

\subsection*{Intersection of $\mathcal{C}$ and $\mathcal{S}$}
In the previous subsection, we use the first implicit equation to obtain properties of the curve $\mathcal{C}$ and we explain that this equation is valid until the point where the curve $\mathcal{C}$ intersects $\mathcal{S}$ (i.e. when $\cos\theta$ changes sign). We now describe under which conditions these two curves intersect.

\begin{figure}[h!]
  \centering
    \includegraphics[scale=.5]{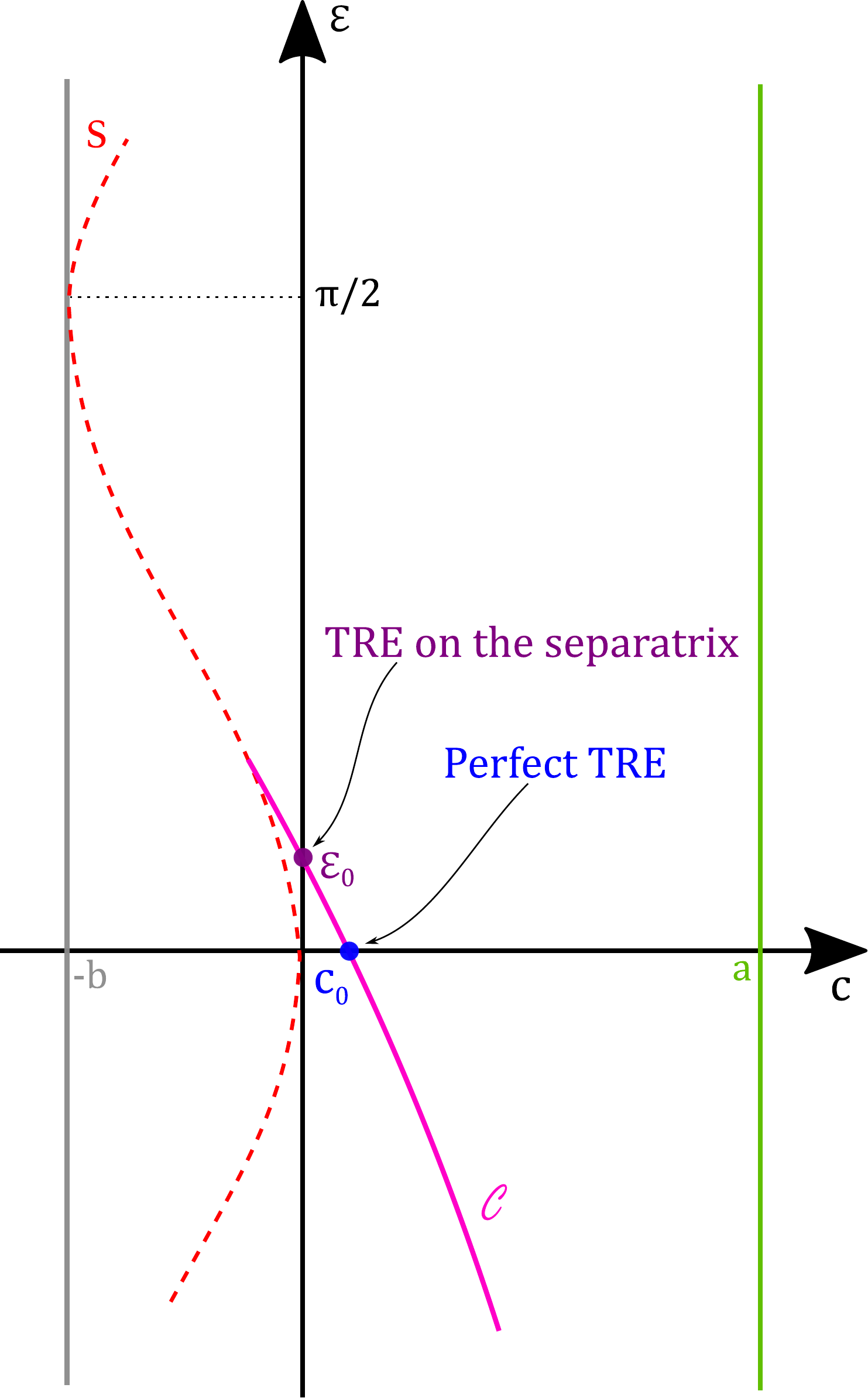}
    \caption{Schematic representation in the space $(c,\varepsilon)$ of the part of the curve $\mathcal{C}$ (solid magenta line) described by the first implicit equation. The red dashed line represents the set $\mathcal{S}$. The grey and green vertical lines have respectively the equations $c=-b$ and $c=a$. The points $\varepsilon_0$ and $c_0$ indicate the positions of a TRE on the separatrix and of a perfect TRE (see the text for details).\label{fig.Solution}}
\end{figure}

First, notice that in the region $b\sin^2(\varepsilon)+c<0$ the integrand of the integral of Eq.~\eqref{eq.1} is not well defined and the curve $\mathcal{S}$ is delimiting this region. If there exists an intersection point between $\mathcal{C}$ and $\mathcal{S}$, then $c$ has to be negative, which corresponds to a positive value of $\varepsilon$ since in this case $\mathcal{C}$ is a decreasing curve, see Fig.~\ref{fig.Solution}. Thus, we can use Eq.~\eqref{eq1}.

\begin{lema}\label{lema.inter}
The curves $\mathcal{C}$ and $\mathcal{S}$ intersect in $(0,1]$, if and only if $\mathcal{I}\leq 2$.
\end{lema}
\begin{proof}
If $\mathcal{C}$ and $\mathcal{S}$ intersect in $(0,1]$, then there exists $u^*\in(0,1]$ such that $F_{a,b}(-bu^*,u^*)=2\pi$. Using Lemmas~\ref{lema.decreasing} and~\ref{lema.F} of~\ref{appa}, we deduce that $u=1$ is the minimum of the function $F_{a,b}(-bu,u)$. Thus $F_{a,b}(-b,1)\leq F_{a,b}(-bu^*,u^*)$,
which implies that $\mathcal{I}\leq 2$.

Conversely, an intersection point between the two curves is characterized by the equation $F_{a,b}(c,u)=2\pi$ with $c=-bu$. The function $F_{a,b}(-bu,u)$ is a continuous function which goes to $\infty$, when $u$ goes to 0 and verifies
\begin{equation*}
F_{a,b}(-bu,u)\to \pi\mathcal{I} \leq 2\pi,
\end{equation*}
when $u$ goes to 1 (using Lemma~\ref{lema.F}). We conclude that there exists $u^*\in(0,1]$, such that $F(-bu^*,u^*)=2\pi$, and thus, an intersection point.
\end{proof}
Notice that $u\in(0,1]$ is equivalent to $\varepsilon\in(0,\frac{\pi}{2}]$, and that the intersection point is unique since $F_{a,b}(-bu,u)$ is a decreasing function (See Lemma~\ref{lema.decreasing}).

\begin{rmk}
Using the preceding proof, we deduce that the intersection point is exactly in $u=1$, if and only if $\mathcal{I}=2$.
\end{rmk}

We derived a condition on the parameters that ensures the existence of an intersection point between $\mathcal{C}$ and $\mathcal{S}$. After this intersection point the implicit equation describing the solution curve $\mathcal{C}$ changes as described below. We first recall that the first implicit equation describes the curve $\mathcal{C}$ in the region $\theta\in[0,\frac{\pi}{2})$, i.e. $\cos\theta> 0$.
Notice that, if there exists an intersection point between $\mathcal{C}$ and $\mathcal{S}$, then $c<0$, which corresponds to oscillating trajectories.
The period $T_\phi$ of these closed orbits can be calculated using the relation
\begin{equation*}
\frac{d\phi}{d\psi}=\pm\frac{1-b\cos^2(\psi)}{\sqrt{(b\cos^2(\psi)+a)(b\cos^2(\psi)+c)}}.
\end{equation*}
Indeed, since the solutions are symmetric and we consider the whole orbit to compute the period, one arrives at
\begin{equation}\label{eq.period}
T_{\phi}(c)=2\int_{-\frac{\pi}{2}+\varepsilon^*}^{\frac{\pi}{2}-\varepsilon^*}\frac{1-b\cos^2(\psi)}{\sqrt{(b\cos^2(\psi)+a)(b\cos^2(\psi)+c)}}d\psi=2F_{a,b}(-bu,u).
\end{equation}
The last equality in Eq.~\eqref{eq.period} is obtained by performing the change of coordinates that we have used before, $x=\cos^2(\psi)$ and $c=-bu=-b\sin^2(\varepsilon)$. Lemma~\ref{lema.decreasing} proves that the period $T_{\phi}(c)$ decreases when $c$ decreases. In Fig.~\ref{fig.PS}, this means that the smaller orbits have a smaller period ($\phi$ plays the role of time in this phase space). This latter result implies that the sign of $\cos(\theta)$ changes after the intersection point between $\mathcal{C}$ and $\mathcal{S}$, since for the values of $c$ between $-b$ and the value of this intersection point the periods are smaller.


The second implicit equation takes the following form
\begin{equation}\label{eq.SIE}
2\pi=\int_{-\pi/2+\varepsilon}^{-\pi/2+\varepsilon^*}-h_{a,b,c}(\psi)d\psi+\int_{-\pi/2+\varepsilon^*}^{\pi/2-\varepsilon^*}h_{a,b,c}(\psi)d\psi+\int_{\pi/2-\varepsilon^*}^{\pi/2-\varepsilon}-h_{a,b,c}(\psi)d\psi ,
\end{equation}
where $\varepsilon^*\in(0,\frac{\pi}{2})$ and
\begin{equation*}
h_{a,b,c}(\psi)=\frac{1-b\cos^2(\psi)}{\sqrt{(a+b\cos^2(\psi))(c+b\cos^2(\psi))}}.
\end{equation*}
Notice that $\varepsilon\in[\varepsilon^*,\pi-\varepsilon^*]$. Hence $\varepsilon^*<\varepsilon$.
Straightforward computation leads to the following expressions in terms of $F_{a,b}$
\begin{equation}\label{eq.SIE1}
2 F_{a,b}(c,-\frac{c}{b})-F_{a,b}(c,u)=2\pi ,
\end{equation}
for $\varepsilon<\frac{\pi}{2}$ and
\begin{equation}\label{eq.SIE2}
2 F_{a,b}(c,-\frac{c}{b})+F_{a,b}(c,u)=2\pi ,
\end{equation}
for $\varepsilon\geq\frac{\pi}{2}$.

After each intersection point between the curves $\mathcal{C}$ and $\mathcal{S}$, the implicit equation describing the curve $\mathcal{C}$ changes as it is explicitly shown in~\ref{ap.GIE}.

\section{The curve $\mathcal{C}$ in the physical and non-physical cases}\label{secphysnonphys}
In this section, we prove that the physical nature of a rigid body can be detected from the structure of the curve $\mathcal{C}$, in the sense that the properties of the curve  $\mathcal{C}$ differs between a physical and a non-physical rigid body.

\subsection{The physical case}\label{sub.phy}
The proof of Th.~\ref{teo.Rinj} is given at the end of this subsection. First, we show two results on which the proof of Th.~\ref{teo.Rinj} is based.

\begin{lema}\label{lema.c}
There exists an intersection point $(c_1,\varepsilon(c_1))$ between $\mathcal{C}$ and $\mathcal{S}$ and a point $(c^*,\frac{\pi}{2})\in\mathcal{C}$ such that $c^*\in(c_2,c_1)$, if and only if $\mathcal{I}<1$. Here $c_2$ is either the coordinate of the second intersection point, or equal to $-b$ if there is no second intersection point.
\end{lema}

\begin{proof}
Since the point $(c^*,\frac{\pi}{2})\in\mathcal{C}$ is found between the first and the second intersection point (in the sense that $c^*\in(c_2,c_1)$), then the implicit equation describing $\mathcal{C}$ in this point is given by Eq.~\eqref{eq.SIE}, leading for $c=c^*$ and $\varepsilon=\frac{\pi}{2}$ to the following expression
\begin{equation*}
2\int_{-\pi/2+\varepsilon^*}^{\pi/2-\varepsilon^*}\frac{1-b\cos^2(\psi)}{\sqrt{(a+b\cos^2(\psi))(c^*+b\cos^2(\psi))}}=2\pi ,
\end{equation*}
which is equivalent to
\begin{equation*}
F_{a,b}(-bu^*,u^*)=\pi,
\end{equation*}
where $u^*=-\frac{c^*}{b}$. Since the minimum value of $F_{a,b}(-bu,u)$ is $\mathcal{I}\pi$ and $c^*\neq -b$, we arrive at $\pi\mathcal{I}<\pi$ concluding that $\mathcal{I}<1$.

Conversely, the condition $\frac{1-b}{\sqrt{b(a+b)}}<1<2$ implies the existence of an intersection point on the interval $u\in(0,1)$ (due to Lemma~\ref{lema.inter}), which leads to $-b<c_1$. Let us put $u_1=-\frac{c_1}{b}$. As a first step, we prove that the condition $\mathcal{I}<1$ gives the existence of a solution $u^*\in(u_1,1)$ to the implicit equation $F_{a,b}(-bu,u)=\pi$. Using the same analysis as the one done in Lemma~\ref{lema.inter}, we observe that $F_{a,b}(-bu,u)$ is equal to $2\pi$ for $u$ equal to $u_1$, and goes to $\pi\mathcal{I}$ when $u$ goes to 1. By hypothesis, $\pi\mathcal{I}<\pi$, we conclude that there exists $u^*\in(u_1,1)$ such that $F_{a,b}(-bu^*,u^*)=\pi$.

Now, we define $c^*=-bu^*<c_1$ and prove that the point $(c,\varepsilon)=(c^*,\frac{\pi}{2})$ belongs to $\mathcal{C}$. We consider expression~\eqref{eq.SIE}, since we analyze $\mathcal{C}$ after the first intersection point.
\begin{equation*}
\int\limits_{0}^{-\pi/2+\varepsilon^*}-h_{a,b,c^*}(\psi)d\psi+\int\limits_{-\pi/2+\varepsilon^*}^{\pi/2-\varepsilon^*}h_{a,b,c^*}(\psi)d\psi+\int\limits_{\pi/2-\varepsilon^*}^{0}-h_{a,b,c^*}(\psi)d\psi=2\int\limits_{-\pi/2+\varepsilon^*}^{\pi/2-\varepsilon^*}h_{a,b,c^*}(\psi)d\psi.
\end{equation*}
Using the change of coordinates $x=\cos^2(\psi)$, we get
\begin{equation*}
2\int_{-\pi/2+\varepsilon^*}^{\pi/2-\varepsilon^*}h_{a,b,c^*}(\psi)d\psi=2F_{a,b}(c^*,\frac{-c^*}{b})=2F_{a,b}(-bu^*,u^*)=2\pi.
\end{equation*}
Thus, $(c^*,\frac{\pi}{2})$ belongs to $\mathcal{C}$. Since it is a solution of the second implicit equation, we also get $c^*\in(c_2,c_1)$, where $c_2$ is either the coordinate of the second intersection point, or $-b$, if there is no second intersection point.
\end{proof}

Following the preceding proof, the same result can be established when $c^*=-b$ and $\mathcal{I}=1$.

\begin{prop}\label{prop.inj}
The function $\varepsilon(c)$ is not injective if and only if there exists a point $(c^*,\frac{\pi}{2})\in\mathcal{C}$, with $c^*\in(-b,0)$.
\end{prop}

\begin{proof}
As the function $\varepsilon(c)$ is not injective, there exists an intersection point $(c_1,\varepsilon(c_1))$ between $\mathcal{C}$ and $\mathcal{S}$; this is because we proved that the function $\varepsilon(c)$ is injective in the part described by the first implicit equation.
Since we have an intersection point, we analyze Eq.~\eqref{eq.SIE1} and prove that $\varepsilon(c)$ is an strictly decreasing function when it is described by this last equation. Let us consider $c_3<c_2<c_1<0$. The period $T_\phi(c)$ being decreasing when $c$ decreases implies $2F_{a,b}(c_3,\frac{-c_3}{b})<2F_{a,b}(c_2,\frac{-c_2}{b})$. Thus, we have
\begin{equation}\label{eq.value}
2F_{a,b}\left(c_3,\frac{-c_3}{b}\right)-2\pi<2F_{a,b}\left(c_2,\frac{-c_2}{b}\right)-2\pi .
\end{equation}
Using $c_3<c_2$, we get
\begin{equation}\label{eq.fu}
0<\frac{1-bx}{\sqrt{x(1-x)(1+\frac{b}{a}x)(x+\frac{c_2}{b})}}<\frac{1-bx}{\sqrt{x(1-x)(1+\frac{b}{a}x)(x+\frac{c_3}{b})}} .
\end{equation}
Thus, if $u_2$ and $u_3$ are such that
\begin{align*}
\frac{1}{\sqrt{ab}}\int_{u_2}^1\frac{1-bx}{\sqrt{x(1-x)(1+\frac{b}{a}x)(x+\frac{c_2}{b})}}&=2F_{a,b}\left(c_2,\frac{-c_2}{b}\right)-2\pi,\\
\frac{1}{\sqrt{ab}}\int_{u_3}^1 \frac{1-bx}{\sqrt{x(1-x)(1+\frac{b}{a}x)(x+\frac{c_3}{b})}}&=2F_{a,b}\left(c_3,\frac{-c_3}{b}\right)-2\pi ,
\end{align*}
using Eq.~\eqref{eq.value} and~\eqref{eq.fu}, we obtain $u_2<u_3$, which is equivalent to
\begin{equation}
\varepsilon(c_2)=\arcsin\left(\sqrt{u_2}\right)<\varepsilon(c_3)=\arcsin\left(\sqrt{u_3}\right).
\end{equation}
We just have proved that when the function $\varepsilon(c)$ is described by Eq.~\eqref{eq.SIE1}, it is a strictly decreasing function. For this reason $\varepsilon(c)$ is injective in the region where it is described by the first implicit equation and by Eq.~\eqref{eq.SIE1}. Since $\varepsilon(c)$ is not injective globally we conclude that there exists a point $(c^*,\frac{\pi}{2})\in\mathcal{C}$, with $c^*\in(-b,0)$.

Conversely, we have that $\varepsilon(c^*)=\frac{\pi}{2}$ and, notice that, for the TRE on the separatrix we have $\varepsilon(0)\in(0,\frac{\pi}{2})$. Thus, all the values in $[\varepsilon(0),\frac{\pi}{2}]$ are taken by the function $\varepsilon(c)$. If the function $\varepsilon(c)$ is lower than $\frac{\pi}{2}$ for values of $c$ in $(-b,c^*)$, then the function is not injective. If it is not the case, then there exists $\delta_1>0$ and $\delta_2>0$ such that
\begin{equation*}
(\frac{\pi}{2}-\delta_2,\frac{\pi}{2}+\delta_2)\subset\varepsilon\left((c^*-\delta_1,c^*+\delta_1)\right).
\end{equation*}
We also have
\begin{equation*}
\varepsilon(c)\in\left[\arcsin\left(\sqrt{\frac{-c}{b}}\right),\pi-\arcsin\left(\sqrt{\frac{-c}{b}}\right)\right] .
\end{equation*}
Since $\arcsin\left(\sqrt{\frac{-c}{b}}\right)\to\frac{\pi}{2}$, when $c\to-b$, then there exists $\tilde{c}$ such that $\tilde{c}<c^*-\delta_1$ and close enough to $-b$ so that
\begin{equation*}
\left[\arcsin\left(\sqrt{\frac{-\tilde{c}}{b}}\right),\pi-\arcsin\left(\sqrt{\frac{-\tilde{c}}{b}}\right)\right]\subset\left(\frac{\pi}{2}-\delta_2,\frac{\pi}{2}+\delta_2\right).
\end{equation*}
Finally, we conclude that $\varepsilon(\tilde{c})\in(\frac{\pi}{2}-\delta_2,\frac{\pi}{2}+\delta_2)$, \textit{i.e.}, the function $\varepsilon(c)$ is not injective.
\end{proof}

\begin{proof} [Proof of Theorem \ref{teo.Rinj}]
If the function $\varepsilon(c)$ is not injective, then there exists an intersection point $(c_1,\varepsilon(c_1))$ between $\mathcal{C}$ and $\mathcal{S}$ and a point $(c^*,\frac{\pi}{2})\in\mathcal{C}$, where $c^*\in(-b,c_1)$. This is a consequence of the proof of Proposition~\ref{prop.inj}. Using Lemma~\ref{lema.c}, we obtain that $\mathcal{I}<1$.
Conversely, if $\mathcal{I}<1$ using Lemma~\ref{lema.c}, we deduce that there exists an intersection point $(c_1,\varepsilon(c_1))$ between $\mathcal{C}$ and $\mathcal{S}$ and a point $(c^*,\frac{\pi}{2})\in\mathcal{C}$, where $c^*\in(-b,c_1)$ that implies that the function $\varepsilon(c)$ is not injective (due to Proposition~\ref{prop.inj}). The first claim follows.

The second claim now follows directly from Proposition~\ref{prop.physical}.
\end{proof}

\subsection{The curve $\mathcal{C}$ in the non-physical case}
Theorem~\ref{teo.Rinj} shows that for physical rigid bodies,  the curve $\mathcal{C}$ is described by an injective decreasing function, while, for non-physical systems, this function is non-injective. In this latter case, the function has an oscillating behavior with respect to the curve $\mathcal{S}$, i.e. the curves $\mathcal{C}$ and $\mathcal{S}$ intersect several times. Moreover, the number of intersection points depends on the value of the geometric constant $\mathcal{I}$, the following Theorem states this dependence. Numerical calculations illustrate these different behaviors in Fig.~\ref{figure2}.


\begin{figure}[h!]
  \centering
  \includegraphics[width=0.8\linewidth]{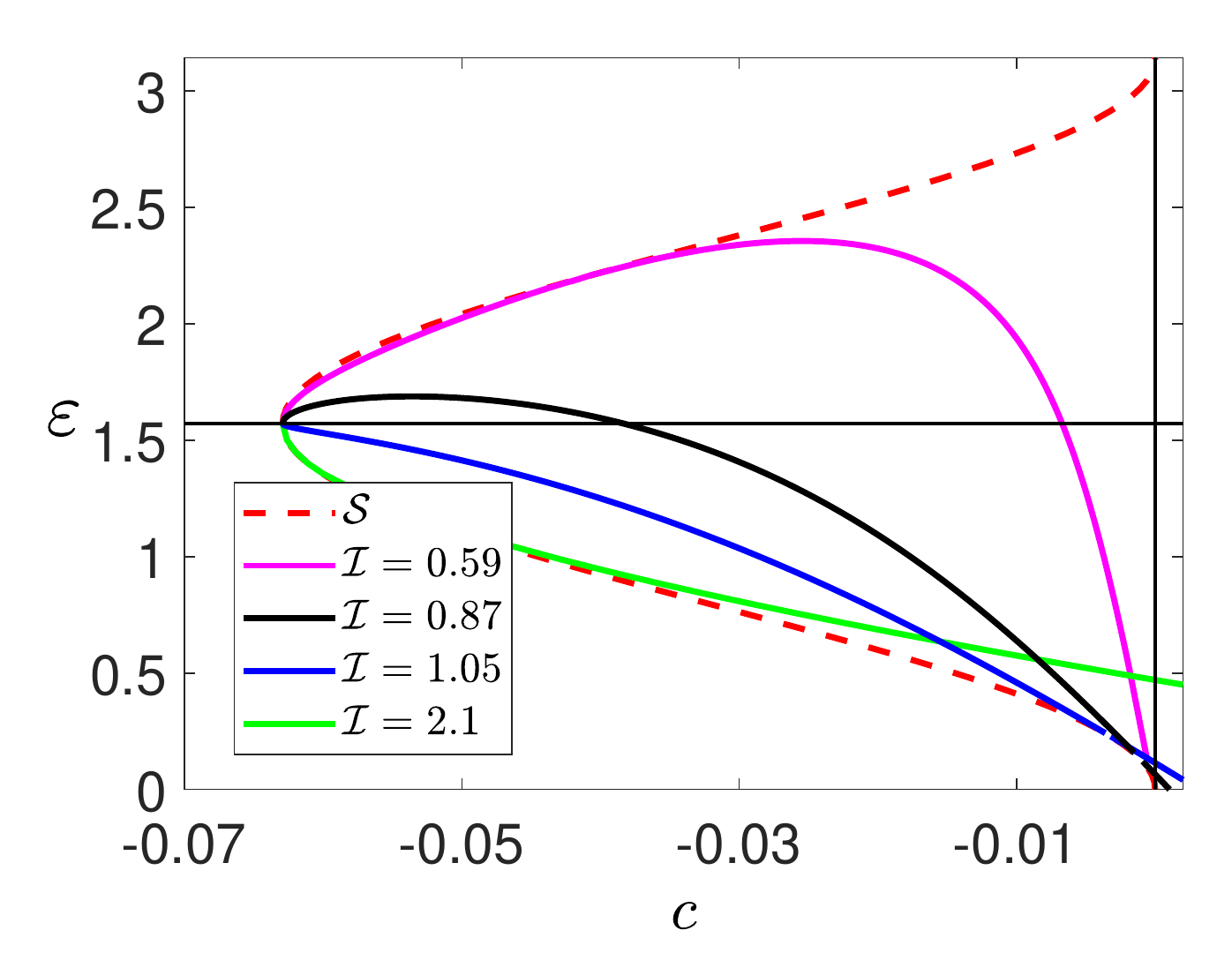}
    \caption{Plot of the function $\varepsilon(c)$. We have used a fixed value of $b=0.0629$, and varied the value of $a$. The solid purple, black, blue and green curves correspond respectively to $a=39.95, a=18.27$, $a=12.65$ and $a=3.1$, which leads to $\mathcal{I}=0.59$, $0.87$, 1.05 and 2.1. The dahsed red curve represents the set $\mathcal{S}$. The black horizontal and vertical lines have the equations $\varepsilon=\pi/2$ and $c=0$, respectively.
\label{figure2}}
\end{figure}

\begin{teo}\label{teo.general}
For every $n\in\mathbb{N}\setminus\{0\}$, if the parameters $a$ and $b$ are such that $\mathcal{I}<\frac{2}{2n+1}$, then there exist at least $n$ intersection points between $\mathcal{C}$ and $\mathcal{S}$. Moreover, if $(c_n,\varepsilon(c_n))$ are the coordinates of the $n$-th intersection point, then there exists $c^*\in(-b,c_n)$ such that $\varepsilon(c^*)=\frac{\pi}{2}$ and another intersection point $(c_{int},\varepsilon(c_{int}))$ such that $c_{int}\in(-b,c_n)$.
\end{teo}

 The proof generalizes the ideas used in the preceding Lemmas and is given in~\ref{ap.GIE}.



\section{Global behavior of the tennis racket effect with respect to the geometric parameters $a$ and $b$}\label{sec.parameters}
\subsection*{Evolution of $c_0$ and $\varepsilon_0$ as a function of $a$}\label{sub.achange}
Recall that $c_0$ and $\varepsilon_0$, introduced in Proposition~\ref{prop.perf}, represent respectively a perfect TRE and the TRE on the separatrix. In this section, we prove that the values of $c_0(a,b)$ and $\varepsilon_0(a,b)$ decrease when $a$ increases. This result allows to analyze the behavior of $\mathcal{C}$ when the parameter $a$ varies. Since we consider different values of $(a,b)$, in this section we use the notation $\mathcal{C}_{a,b}$ instead of $\mathcal{C}$.

\begin{lema}
Let us consider two values of the parameter $a$ such that $0<a_1<a_2$ and the solutions $(0,c_0(a_1))\in\mathcal{C}_{a_1,b}$, $(0,c_0(a_2))\in\mathcal{C}_{a_2,b}$. Then $c_0(a_1)>c_0(a_2)$.
\end{lema}
\begin{proof}
We have that $(0,c_0(a_1))\in\mathcal{C}_{a_1,b}$ and $(0,c_0(a_2))\in\mathcal{C}_{a_2,b}$. Thus
\begin{equation*}
\pi=\int_{0}^{\pi/2} \frac{(1-b\cos^2(\psi))d\psi}{\sqrt{(b\cos^2(\psi)+a_1)(b\cos^2(\psi)+c_0(a_1))}}=\int_{0}^{\pi/2} \frac{(1-b\cos^2(\psi))d\psi}{\sqrt{(b\cos^2(\psi)+a_2)(b\cos^2(\psi)+c_0(a_2))}}.
\end{equation*}
Since $a_1<a_2$, then $\sqrt{b\cos^2(\psi)+a_1}<\sqrt{b\cos^2(\psi)+a_2}$. This latter inequality implies
\begin{equation*}
\frac{1-b\cos^2(\psi)}{\sqrt{b\cos^2(\psi)+a_2}}<\frac{1-b\cos^2(\psi)}{\sqrt{b\cos^2(\psi)+a_1}} .
\end{equation*}
From the two previous equations, we obtain
\begin{equation*}
\frac{1}{\sqrt{b\cos^2(\psi)+c_0(a_2)}}>\frac{1}{\sqrt{b\cos^2(\psi)+c_0(a_1)}},
\end{equation*}
which leads to
\begin{equation*}
c_0(a_1)>c_0(a_2).
\end{equation*}
\end{proof}

\begin{lema}\label{lema.a8}
Let us consider two values of the parameter $a$ such that $0<a_1<a_2$ and the solutions $(\varepsilon_0(a_1),0)\in\mathcal{C}_{a_1,b}$, $(\varepsilon_0(a_2),0)\in\mathcal{C}_{a_2,b}$. Then $\varepsilon_0(a_1)>\varepsilon_0(a_2)$.
\end{lema}
\begin{proof}
We follow the same idea as in the previous proof. The fact that $(\varepsilon_0(a_1),0)\in\mathcal{C}_{a_1,b}$ and $(\varepsilon_0(a_2),0)\in\mathcal{C}_{a_2,b}$ leads to
\begin{equation*}
\pi=\int_{0}^{\pi/2-\varepsilon_0(a_1)} \frac{(1-b\cos^2(\psi))d\psi}{\sqrt{(b\cos^2(\psi)+a_1)b\cos^2(\psi)}}=\int_{0}^{\pi/2-\varepsilon_0(a_2)} \frac{(1-b\cos^2(\psi))d\psi}{\sqrt{(b\cos^2(\psi)+a_2)b\cos^2(\psi)}}.
\end{equation*}
Since $a_1<a_2$, then $\sqrt{(b\cos^2(\psi)+a_1)b\cos^2(\psi)}<\sqrt{(b\cos^2(\psi)+a_2)b\cos^2(\psi)}$, giving
\begin{equation*}
\frac{1-b\cos^2(\psi)}{\sqrt{(b\cos^2(\psi)+a_2)b\cos^2(\psi)}}<\frac{1-b\cos^2(\psi)}{\sqrt{(b\cos^2(\psi)+a_1)b\cos^2(\psi)}} .
\end{equation*}
From the two previous equations, we get
\begin{equation*}
\frac{\pi}{2}-\varepsilon_0(a_1)<\frac{\pi}{2}-\varepsilon_0(a_2).
\end{equation*}
We conclude that $\varepsilon_0(a_1)>\varepsilon_0(a_2)$.
\end{proof}
All these properties are verified for any allowed value of $a$ and $b$. Notice that, when we change the value of $a$, but not the value of $b$ we have  the same limit curve $\mathcal{S}_b$ that cannot be crossed, but the set $\mathcal{C}_{a,b}$ is modified. According to these results, we know that when the value of $a$ increases the curve $\mathcal{C}_{a,b}$ goes down, in the sense that its intersection points with the two axes decrease (See Figs.~\ref{fig.Solution} and~\ref{fig:b}). An illustrative numerical example is given in Fig.~\ref{fig:b} for physical and non-physical rigid bodies. We observe that the evolution of $\varepsilon$ as a function of $c$ is similar in the two cases. When $a$ goes to $\infty$ for a fixed value of $b$, the curve $\mathcal{C}$ is tangent to $\mathcal{S}$ in $c=\varepsilon=0$. This limit case is described in~\cite{mardesic2020}.

\begin{figure}[h!]
    \centering
    \includegraphics[scale=.8]{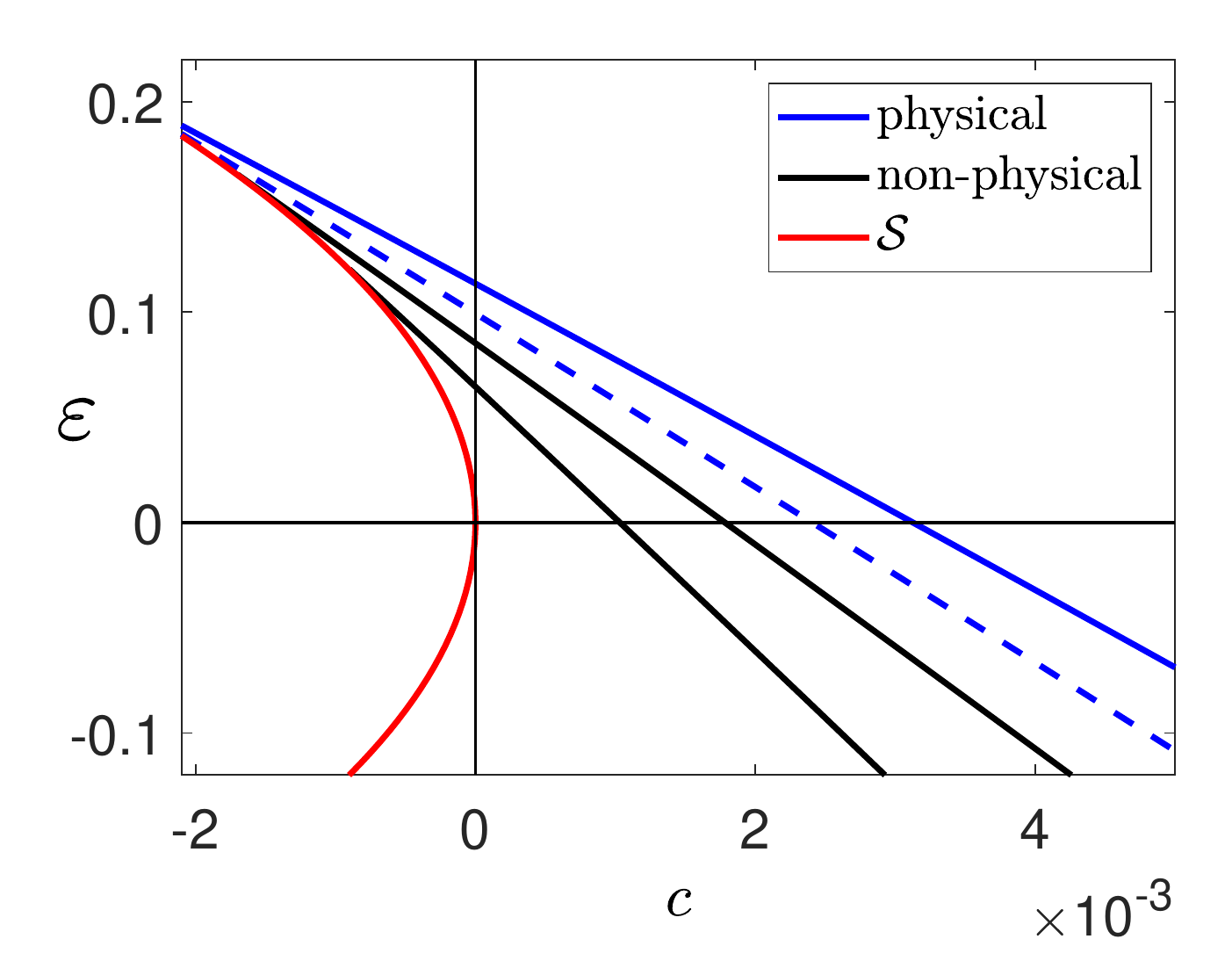}
    \caption{Plot of the defect $\varepsilon$ as a function of $c$ (curve $\mathcal{C}_{a,b}$) for different values of $a$, the parameter $b$ being fixed to 0.06. The parameter $a$ is respectively set to 12.65, 13.89, 15.38 and 18.27 (solid blue, dashed blue and 2 solid black lines), which leads to $\mathcal{I}$ equal to 1.045, 1, 0.951 and 0.872. The horizontal and vertical solid lines indicate respectively the position of a perfect TRE ($\varepsilon=0$) and of the separatrix ($c=0$). The set $\mathcal{S}$ is plotted in red.}
    \label{fig:b}
\end{figure}

\begin{rmk}
The only possibility to have a perfect TRE on the separatrix is to consider the abstract limit $a\to\infty$. This case is described in~\cite{mardesic2020}. Moreover, as explained in Sec.~\ref{sec.const}, the physical constraint corresponds to $2b+ab \leq 1$. If this inequality is fulfilled when $a\to\infty$, then $b\to 0$. This case ($a\to\infty$, $b\to 0$) is described below. Note that when $a\to\infty$ and $b\to 0$, the physical constraint becomes $ab \leq 1$.
\end{rmk}

\subsection*{Behavior of the flip defect $\varepsilon$ when $a\rightarrow \infty$ and $b \rightarrow 0$}\label{sub.limit}
To derive an explicit expression for the flip defect $\varepsilon$ in the limit case $a\rightarrow \infty$ and $b \rightarrow 0$, we start with the implicit equation
\begin{equation}
\label{eqimpli1} \int_{0}^{\pi/2 - \varepsilon} \frac{1- b \cos^2 \psi}{\sqrt{(a+b \cos^2 \psi)(c+b \cos^2 \psi)}}  d \psi = \int_{0}^{\pi/2 - \varepsilon} \frac{1- b \cos^2 \psi}{\sqrt{ab(1+\frac{b}{a} \cos^2 \psi)(\frac{c}{b}+ \cos^2 \psi)}}  d \psi = \pi.
\end{equation}
Therefore, when $a\rightarrow \infty$ and $b \rightarrow 0$, we have
\begin{equation}
\label{eqimpli2} \frac{1}{\sqrt{ab}}\int_{0}^{\pi/2 - \varepsilon} \frac{1}{\sqrt{\frac{c}{b}+ \cos^2 \psi}}  d \psi = \pi.
\end{equation}
If we define the parameters $\chi$ and $\beta$ as follows
\begin{equation}
\label{twis3} \chi = a b, \;\;\;\;\;\;\; \beta =  \frac{c}{b},
\end{equation}
and use the identity $\cos^2 \psi = 1 - \sin^2 \psi$, then Eq. (\ref{eqimpli2}) can be written as
\begin{equation}
\label{eqimpli3} \frac{1}{\sqrt{\chi(1+\beta)}}\int_{0}^{\pi/2 - \varepsilon} \frac{1}{\sqrt{1 - \frac{1}{1+\beta}\sin^2 \psi}}  d \psi= \frac{1}{\sqrt{\chi(1+\beta)}} F\big(\frac{\pi}{2} - \varepsilon\Big{|}\frac{1}{1+\beta}\big) = \pi,
\end{equation}
where the function $F$ is the incomplete elliptic integral of the first kind defined by
\begin{equation}
\label{twis2} F(\varphi|m) = \int_{0}^{\varphi}
\frac{d\theta}{\sqrt{1-m \sin ^2\theta }} = \int_{0}^{\sin \varphi} \frac{1}{\sqrt{(1-t^2)(1-m t^2)}} dt.
\end{equation}
Therefore, from Eq. (\ref{eqimpli3}), we can obtain the flip defect $\varepsilon$
\begin{equation}
\label{twis1} \varepsilon = \frac{\pi}{2}-\text{am}\left(\pi  \sqrt{\chi(1+\beta)} \Big{|}\frac{1}{1+\beta
}\right),
\end{equation}
where am is the Jacobi amplitude.
Recall that the Jacobi amplitude is defined as the
inverse of the incomplete elliptic integral of the first kind, i.e. if $ u =  F(\varphi|m) $, then $\varphi =$am$(u
|m)$.

According to Lemma~\ref{lema.a8}, we know that the value of the flip defect on the separatrix $\varepsilon_0$ decreases when $a$ increases, and so the smallest possible value of $\varepsilon_0$ will occur when $a\rightarrow \infty$. Therefore, we can use expression~(\ref{twis1}) to determine the minimum value of $\varepsilon_0$ subject to the physical constraint $\chi \leq 1$. It turns out that in the case of $c<0$, there is a
minimum value for $\varepsilon$ when $\chi \rightarrow
1$ and $\beta \rightarrow 0$. In this limit case, we have
\begin{equation}
\label{twis4} \varepsilon_{0,min} = \frac{\pi }{2}-\text{am}(\pi|1)
= \pi -2 \arctan\left(e^{\pi }\right) = 0.086374.
\end{equation}
While in the case $c > 0$, there is
a minimum value of $\beta \neq 0$ such that $\varepsilon = 0$ that happens
when $\chi \rightarrow 1$, namely
\begin{equation}
\label{twis5} \frac{\pi}{2}-\text{am}\left(\pi  \sqrt{(1+\beta)} \Big{|}\frac{1}{1+\beta
}\right) =\, 0
 \;\;\; \Rightarrow \;\;\; \beta = 0.028973.
\end{equation}

\begin{figure}[h!]
  \centering
  \includegraphics[width=0.8\linewidth]{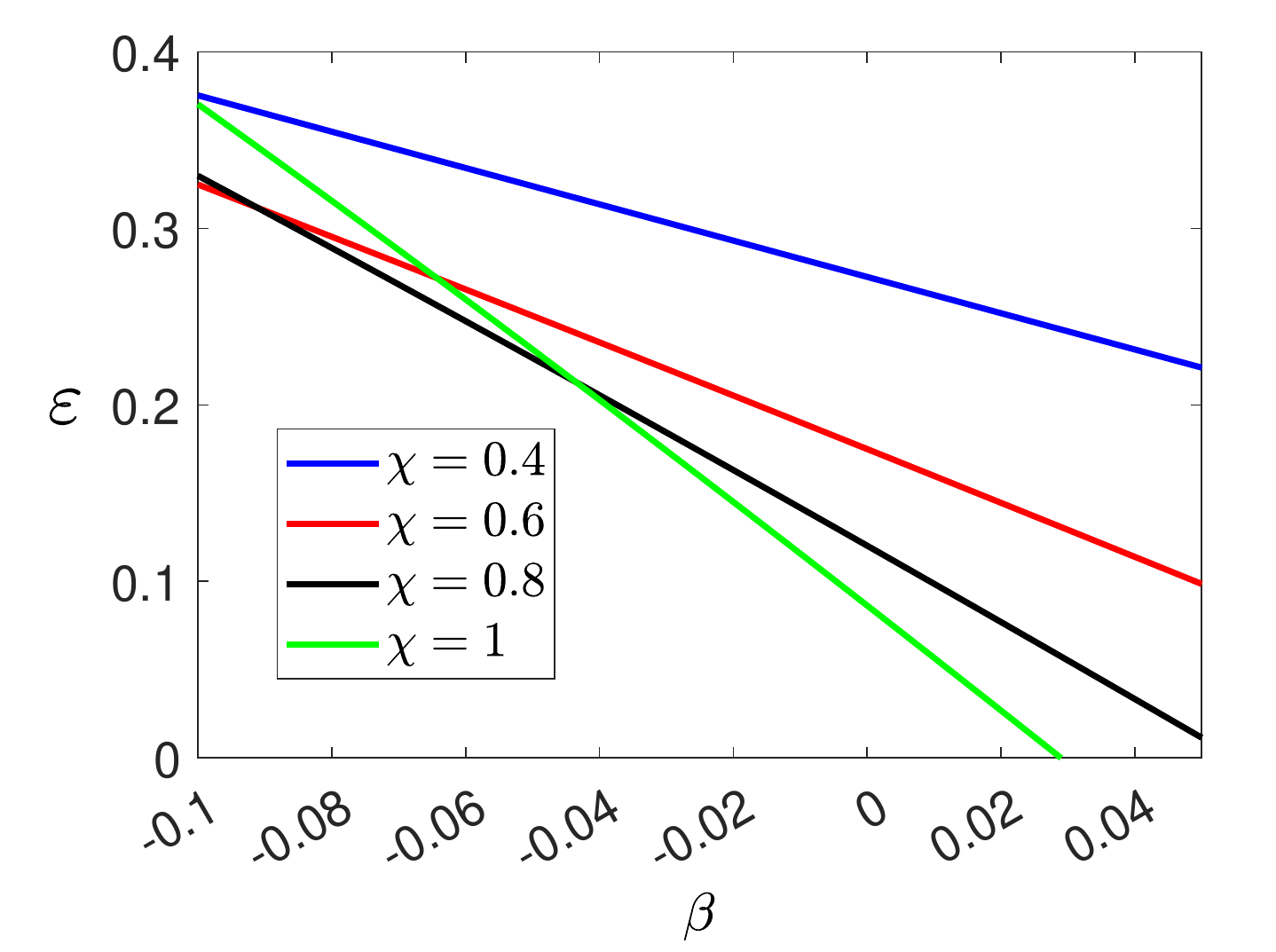}
    \caption{Plot of the function $\varepsilon = \frac{\pi}{2}-\text{am}\left(\pi  \sqrt{\chi(1+\beta)} \Big{|}\frac{1}{1+\beta
}\right) $. The values of $\chi=1$, $0.8$, $0.6$ and $\chi=0.4$ correspond to the green, black, red and blue curves respectively. \label{figureasym1}}
\end{figure}
Figure~\ref{figureasym1} illustrates the analysis carried out above. For instance, we observe that when we are very
close to the separatrix $\beta = 0$, there is no physical rigid
body (in the sense that $\chi \leq 1$) that can achieve a perfect $\Delta \psi =
\pi$ twist (i.e., $\varepsilon_0=0$). Nevertheless, we observe that when
$\beta = 0$, the smallest possible value of $\varepsilon_0$ happens
for $\chi \rightarrow 1$, and therefore
such a physical rigid body characterized by $\chi=1$ will perform the best quasi-perfect twist, $\Delta \psi =\pi-2\varepsilon_{0,min}=2.96884$, on the separatrix.


\section{Conclusion}\label{sec:conc}
In conclusion, we show in this study how the physical or non-physical nature of a rigid body can be detected in the behavior of two geometric effects, namely the Montgomery Phase and the Tennis Racket Effect, characterizing the free rotational dynamics of an asymmetric body. We establish that the signatures of the physical constraints are geometric, i.e. they correspond to qualitative changes of the two geometric effects. The existence of these signatures is in itself a remarkable and unexpected result. We also introduce two natural geometric constants, $\mathcal{I}$ and $\mathcal{J}$, depending on the moments of inertia of the body.  We show that the geometric signatures can only be characterized  in terms of these two constants. For instance, we prove that the evolution of the defect to a perfect TRE with respect to the energy of the system is described by an injective function if and only if $\mathcal{I}>1$, i.e. in the case where the physical constraints are satisfied. We illustrate these findings with numerical simulations to show the generic character of these signatures.

Our study opens the way to different problems to be explored. A first question concerns the experimental detection of the physical constraints. We establish a series of results regarding the geometric properties of the rigid body. An interesting problem would be to observe experimentally these geometric signatures. Such measurements can be performed quite simply for standard rigid bodies such as a mobile phone~\cite{phone2021}. Following the results established in~\cite{mardesic2020}, another open issue is to find geometric signatures in the Monster Flip effect, a freestyle skateboard trick which can be observed in the rotational dynamics of any asymmetric rigid body. Another intriguing problem is the role of such physical constraints in the quantum domain~\cite{zareBook,landauBookMQ,RMPsugny}. In line with this study, a first point to explore is the geometric signatures of the physical constraints on the quantum dynamical properties of the system. A possible signature is based on the quantum rotation number, a characteristic of the structure of the molecular spectrum, which is closely related to the Montgomery phase as shown in~\cite{vandamme2017b}. Attention must be paid to the rotational constants and their connection with the moments of inertia. An example is given by the water molecule for which the rotational constants for the vibrational ground state are given by $A=9.28$~cm$^{-1}$, $B=14.52$~cm$^{-1}$ and $C=27.88$~cm$^{-1}$~\cite{pawlowski2002}. However, only the equilibrium rotational constants are inversely proportional to the moments of inertia, with e.g. $A_e=h/(8\pi^2cI_A)$. In our example, we have $A_e=9.5077$~cm$^{-1}$, $B_e=14.5924$~cm$^{-1}$ and $C_e=27.2660$~cm$^{-1}$ leading to $\mathcal{I}=1.001$, which is approximately (up to experimental uncertainties) equal to 1 as expected for a planar classical molecule. A direct calculation from the first set of rotational constants leads to $\mathcal{I}=0.946$. The difference between the two values of rotational constants is due to vibrational corrections~\cite{pawlowski2002}. In this case, the quantum system cannot be considered as rigid which explains why the geometric constant is less than one. Finally, it would also be relevant to study the role of such physical constraints in the control of classical or quantum rigid bodies~\cite{RMPsugny,pozzoli2022,Pozzoli2022b,babilotte2016}.\\ \\
\noindent \textbf{Acknowledgments.} This work was supported by the Ecole Universitaire de Recherche (EUR)-EIPHI Graduate School (Grant No. 17- EURE-0002). We thank V. Boudon for helpful discussions about the quantum case. This research has been supported by the ANR-DFG
project ‘CoRoMo’.

\appendix
\section{The function $F_{a,b}(-bu,u)$}\label{appa}
In this appendix, we prove the properties of the function $F_{a,b}(-bu,u)$ that are used in the main text.

\begin{lema}\label{lema.decreasing}
The function $F_{a,b}(-bu,u)$ is a decreasing function on the interval $[0,1]$ and thus, its minimum value is taken in $u=1$.
\end{lema}
\begin{proof}
Recall that the function $F_{a,b}(-bu,u)$ has been defined as
\begin{align}
\label{conn1} F_{a,b}(-bu,u)=\frac{1}{\sqrt{ab}}\int_u^1 \frac{1-bx}{\sqrt{x(1-x)(1+\frac{b}{a}x)(x-u)}}dx .
\end{align}
Let us perform the following change of variables
\begin{align}
\label{con2} x = u + (1 - u) y^2 ,
\end{align}
so for $x \in (u,1)$, in terms of $y$, we have $ y \in (0,1)$.
Therefore the integral (\ref{conn1}) can be written as
\begin{align}
\label{con3} F_{a,b}(-b u ,u) = \frac{2}{\sqrt{ab}}\int_0^1
\frac{1-b \left(u+(1-u) y^2\right)}{\sqrt{\left(1-y^2\right)
\left(u+(1-u) y^2\right) \left(1+\frac{b \left(u+(1-u)
   y^2\right)}{a}\right)}}dy.
\end{align}
Since now the limits of integration (\ref{con3}) do not depend
on the variable $u$, to compute the derivative $\partial_u
F_{a,b}(-b u ,u)$, we simply perform the following computation
\begin{align}
 \partial_u F_{a,b}(-b u ,u) &=
\frac{2}{\sqrt{ab}}\int_0^1 \partial_u \Big{[}\frac{1-b
\left(u+(1-u) y^2\right)}{\sqrt{\left(1-y^2\right) \left(u+(1-u)
y^2\right) \left(1+\frac{b \left(u+(1-u)
   y^2\right)}{a}\right)}}\Big{]}dy \nonumber \\
\label{con4} &= -\frac{b}{(a b)^{3/2}}  \int_0^1
\frac{\sqrt{1-y^2} \left(a+(a+2) b \left(u+(1-u)
y^2\right)\right)}{\left(\left(u+(1-u) y^2\right) \left(1+\frac{b
   \left(u+(1-u) y^2\right)}{a}\right)\right)^{3/2}} dy.
\end{align}
Note that we have an explicit minus sign in the right hand side of
Eq.~\eqref{con4}, and since all the quantities inside the integrand
are positive, we conclude that the derivative $\partial_u
F_{a,b}(-b u ,u)$ is always negative. This implies that the
function $F_{a,b}(-bu,u)$ is an injective decreasing function on the
interval $u \in [0,1]$, and so its global minimum
corresponds to the point where $u \rightarrow 1$.
\end{proof}

\begin{lema}\label{lema.F}
The function $F_{a,b}(-bu,u)$ has the following limiting behaviors
\begin{equation}
\lim_{u\to 0}F_{a,b}(-bu,u)=\infty ,
\end{equation}

\begin{equation}
\lim_{u\to 1}F_{a,b}(-bu,u)=\frac{1-b}{\sqrt{b(a+b)}}\pi .
\end{equation}
\end{lema}
\begin{proof}
Given the expression Eq.~\eqref{conn1} of the function $F_{a,b}(-bu,u)$, it is straightforward to deduce that this function diverges when $u\to 0$. The second statement of the Lemma can be shown as follows. When $u$ is close to one, the function $F_{a,b}(-bu,u)$ can be approximated as follows
\begin{equation*}
F_{a,b}(-bu,u)=\frac{1}{\sqrt{ab}}\int_u^1 \frac{1-bx}{\sqrt{x(1-x)(1+\frac{b}{a}x)(x-u)}}dx\sim\frac{1-b}{\sqrt{ab(1+\frac{b}{a})}}\int_u^1 \frac{1}{\sqrt{(1-x)(x-u)}}dx.
\end{equation*}
By computing the second integral, we obtain
\begin{equation}\label{eq.aprox}
\frac{1-b}{\sqrt{ab(1+\frac{b}{a})}}\int_u^1 \frac{1}{\sqrt{(1-x)(x-u)}}dx=\frac{1-b}{\sqrt{ab(1+\frac{b}{a})}}\pi .
\end{equation}
Then we estimate the error $\mathcal{E}$ of this approximation
\begin{align*}
|\mathcal{E}|=&\left|\frac{1}{\sqrt{ab}}\int_u^1 \frac{1-bx}{\sqrt{x(1-x)(1+\frac{b}{a}x)(x-u)}}dx-\frac{1-b}{\sqrt{ab(1+\frac{b}{a})}}\int_u^1 \frac{1}{\sqrt{(1-x)(x-u)}}dx\right|\\
=&\frac{1-b}{\sqrt{ab(1+\frac{b}{a})}}\left|\int_u^1 \frac{1}{\sqrt{(1-x)(x-u)}}\left(\frac{(1-bx)\sqrt{1+\frac{b}{a}}}{(1-b)\sqrt{x(1+\frac{b}{a}x)}}-1\right)dx\right|\\
\leq &\frac{1-b}{\sqrt{ab(1+\frac{b}{a})}}\int_u^1 \frac{dx}{\sqrt{(1-x)(x-u)}}\textit{ max }_{x\in[u,1]}\left|\frac{(1-bx)\sqrt{1+\frac{b}{a}}}{(1-b)\sqrt{x(1+\frac{b}{a}x)}}-1\right|\\
=&\frac{1-b}{\sqrt{ab(1+\frac{b}{a})}}\pi\textit{ max }_{x\in[u,1]}\left|\frac{(1-bx)\sqrt{1+\frac{b}{a}}-(1-b)\sqrt{x(1+\frac{b}{a}x)}}{(1-b)\sqrt{x(1+\frac{b}{a}x)}}\right| ,\\
\end{align*}
and we get that $|\mathcal{E}|\to 0$, when $u\to 1$.

Using this result and Eq.~(\ref{eq.aprox}), we get
\begin{equation*}
F_{a,b}(-bu,u)\to \frac{1-b}{\sqrt{ab(1+\frac{b}{a})}}\pi ,
\end{equation*}
when $u\to 1$.
\end{proof}

\section{General implicit equations in the non-physical case}\label{ap.GIE}
 The same derivation as the one used for the second implicit equation can be done after each intersection point between $\mathcal{C}$ and $\mathcal{S}$, since the period $T_{\phi}(c)$ is a decreasing function. We deduce that the implicit equation describing $\mathcal{C}$ depends on the number of intersection points.

After the first $n$ intersection points, the implicit equation that describes $\mathcal{C}$ is
\begin{equation}\label{eq.odd}
2\pi=\int_{-\pi/2+\varepsilon}^{-\pi/2+\varepsilon^*}-h_{a,b,c}(\psi)d\psi+(2n-1)\int_{-\pi/2+\varepsilon^*}^{\pi/2-\varepsilon^*}h_{a,b,c}(\psi)d\psi+\int_{\pi/2-\varepsilon^*}^{\pi/2-\varepsilon}-h_{a,b,c}(\psi)d\psi ,
\end{equation}
for $n$ odd and
\begin{equation}\label{eq.even}
2\pi=\int_{-\pi/2+\varepsilon}^{\pi/2-\varepsilon^*}h_{a,b,c}(\psi)d\psi+(2n-1)\int_{-\pi/2+\varepsilon^*}^{\pi/2-\varepsilon^*}h_{a,b,c}(\psi)d\psi+\int_{-\pi/2+\varepsilon^*}^{\pi/2-\varepsilon}h_{a,b,c}(\psi)d\psi ,
\end{equation}
for $n$ even. Using $x=\cos^2(\psi)$, we get four different expressions
\begin{eqnarray}
& & 2nF_{a,b}(c,-\frac{c}{b})-F_{a,b}(c,u)=2\pi,~n~\textrm{odd},~\varepsilon<\frac{\pi}{2}\label{eq.odd1}\\
& & 2nF_{a,b}(c,-\frac{c}{b})+F_{a,b}(c,u)=2\pi,~n~\textrm{odd},~\varepsilon\geq\frac{\pi}{2}\label{eq.odd2}\\
& &
2nF_{a,b}(c,-\frac{c}{b})-F_{a,b}(c,u)=2\pi,~n~\textrm{even},~\varepsilon>\frac{\pi}{2} \label{eq.even1}\\
& & 2nF_{a,b}(c,-\frac{c}{b})+F_{a,b}(c,u)=2\pi ,~n~\textrm{even},~\varepsilon\leq\frac{\pi}{2}\label{eq.even2}
\end{eqnarray}

\begin{lema}\label{lema.23}
Assume that the parameters $a$ and $b$ are such that $\mathcal{I}<\frac{2}{3}$. Let $(c_1,\varepsilon(c_1))$ be the coordinates of the first intersection point between $\mathcal{C}$ and $\mathcal{S}$. Then there exists an intersection point $(c_{int},\varepsilon(c_{int}))$ such that $c_{int}\in(-b,c_1)$.
\end{lema}
\begin{proof}
First notice that the condition $\mathcal{I}<\frac{2}{3}<2$ implies the existence of an intersection point on the interval $\varepsilon\in(0,\frac{\pi}{2})$, due to Lemma~\ref{lema.inter}, which gives $-b<c_1$. Then, we stress that  $\mathcal{I}<\frac{2}{3}<1$ leads to the existence of a point $c^*\in(-b,c_1)$, such that $\varepsilon(c^*)=\frac{\pi}{2}$.

If there exists another intersection point between $c^*$ and $c_1$, the proof is done. If not, then we look for one after $c^*$, i.e., on the interval $(-b,c^*)$. For this reason we use the expression~\eqref{eq.odd2}, for $n=1$
\begin{equation*}
2F_{a,b}(c,-\frac{c}{b})+F_{a,b}(c,u)=2\pi .
\end{equation*}
We analyze this equation along the curve $c=-bu$ to find intersection points, getting
\begin{equation*}
2\pi=2F_{a,b}(-bu,u)+F_{a,b}(-bu,u)=3F_{a,b}(-bu,u).
\end{equation*}
Thus, our problem is equivalent to the problem of finding solutions to the implicit equation given by
\begin{equation*}
F_{a,b}(-bu,u)=\frac{2\pi}{3}.
\end{equation*}
Using the same techniques as before, we notice that $F_{a,b}(-bu,u)$ is equal to $\pi$, for $u^*=\frac{c^*}{-b}$ and goes to $\mathcal{I}\pi<\frac{2}{3}\pi$,  when $u$ goes to 1. We deduce that there exists $u_{int}\in(u^*,1)$ such that $F_{a,b}(-bu,u)$ is equal to $\frac{2\pi}{3}$. Then the point $(c_{int},u_{int})$, where $c_{int}=-bu_{int}$, is an intersection point such that $-b<c_{int}<c^*<c_1$.
\end{proof}
\begin{proof}[Proof of Th.~\ref{teo.general}] We proceed by induction
\begin{itemize}
\item Base case.
For $n=1$, we get the condition $\mathcal{I}<\frac{2}{3}<1<2$. Using Lemma~\ref{lema.inter}, we have the existence of at least one intersection point. Using Lemma~\ref{lema.c} and Lemma~\ref{lema.23}, we have respectively the existence of $c^*$ and of the intersection point $(c_{int},\varepsilon(c_{int}))$.

\item Induction step. Assuming that the statement holds for $n>1$, we prove that it holds for $n+1$.
The condition $\mathcal{I}<\frac{2}{2(n+1)+1}<\frac{2}{2n+1}$ implies the existence of at least $n+1$ intersection points (the statement holds for $n$). Let $(c_{n+1},\varepsilon(c_{n+1}))$ be the coordinates of the $(n+1)$-st intersection point and $u_n+1=-\frac{c_{n+1}}{b}$. Notice that $-b<c_{n+1}$ and, hence, $u_{n+1}<1$ (the statement holds for $n$).

First we prove that the condition $\mathcal{I}<\frac{2}{2(n+1)+1}<\frac{1}{n+1}$ implies the existence of a solution $u^*\in(u_{n+1},1)$ to the implicit equation $F_{a,b}(-bu,u)=\frac{\pi}{n+1}$. Using the same analysis as the one done in Lemma~\ref{lema.inter}, we observe that $F_{a,b}(-bu,u)$ is equal to $\frac{2\pi}{2n+1}$, for $u$ equal to $u_{n+1}$, and goes to $\mathcal{I}\pi$, when $u$ goes to 1. Using the hypothesis, we get
\begin{equation*}
\mathcal{I}\pi<\frac{\pi}{n+1}<\frac{2\pi}{2n+1}.
\end{equation*}
We conclude that there exists $u^*\in(u_{n+1},1)$ such that $F_{a,b}(-bu^*,u^*)=\frac{\pi}{n+1}$. Now, we define $-b<c^*=-bu^*<c_{n+1}$ and we show that the point $(c,\varepsilon)=(c^*,\frac{\pi}{2})$ belongs to $\mathcal{C}$. We study the set $\mathcal{C}$ after the first $n+1$ intersection points and we observe that Eq.~\eqref{eq.odd} and~\eqref{eq.even} transform into
\begin{equation*}
2(n+1)\int\limits_{-\pi/2+\varepsilon^*}^{\pi/2-\varepsilon^*}h_{a,b,c^*}(\psi)d\psi,
\end{equation*}
for the point $(c^*,\frac{\pi}{2})$. Using the change of coordinates $x=\cos^2(\psi)$, we get
\begin{equation*}
2(n+1)\int_{-\pi/2+\varepsilon^*}^{\pi/2-\varepsilon^*}h_{a,b,c^*}(\psi)d\psi=2(n+1)F_{a,b}(c^*,\frac{-c^*}{b})=2(n+1)F_{a,b}(-bu^*,u^*)=2\pi.
\end{equation*}
Thus, $(c^*,\frac{\pi}{2})$ belongs to $\mathcal{C}$, which implies that $\varepsilon(c^*)=\frac{\pi}{2}$.

Finally, we prove the existence of the intersection point $(c_{int},\varepsilon(c_{int}))$. If there exists another intersection point between $c^*$ and $c_{n+1}$, the proof is over. If not, then we look for one after $c^*$, i.e., on the interval $(-b,c^*)$. For this reason we use the implicit equation~\eqref{eq.odd2} or~\eqref{eq.even2} along the curve $c=-bu$. Both equations transform into
\begin{equation*}
2\pi=2(n+1)F_{a,b}(-bu,u)+F_{a,b}(-bu,u)=(2(n+1)+1)F_{a,b}(-bu,u).
\end{equation*}
Thus, our problem is equivalent to the problem of finding solutions to the implicit equation given by
\begin{equation*}
F_{a,b}(-bu,u)=\frac{2\pi}{2(n+1)+1}.
\end{equation*}
Using the same techniques as before, we notice that $F_{a,b}(-bu,u)$ is equal to $\frac{\pi}{n+1}$, for $u^*=\frac{c^*}{-b}$ and goes to $\frac{1-b}{\sqrt{b(a+b)}}\pi$, when $u$ goes to 1. We get
\begin{equation*}
\frac{1-b}{\sqrt{b(a+b)}}\pi<\frac{2\pi}{2(n+1)+1}<\frac{\pi}{n+1}.
\end{equation*}
We conclude that there exists $u_{int}\in(u^*,1)$ such that $F_{a,b}(-bu,u)$ is equal to $\frac{2\pi}{2(n+1)+1}$. Then, the point $(c_{int},u_{int})$, where $c_{int}=-bu_{int}$, is an intersection point such that $-b<c_{int}<c^*<c_{n+1}$.
\end{itemize}
\end{proof}
\bibliographystyle{vancouver}

\end{document}